\RequirePackage[british]{babel}
\documentclass[reqno,a4paper,12pt]{amsart}
\usepackage[a4paper,hmargin=2cm,tmargin=3cm,bmargin=3cm]{geometry}
\usepackage[utf8x]{inputenc}
\usepackage{amsmath,amssymb,amstext,amsthm,amscd,mathrsfs,eucal}
\usepackage{graphicx,color}
\usepackage{array}
\usepackage{cite}
\usepackage[vcentermath]{youngtab}
\Yboxdim{4pt}
\usepackage{hyperref}
\hypersetup{%
  pdftitle   = {Metric 3-Lie algebras for unitary Bagger--Lambert theories},
  pdfkeywords = {indefinite, M2, M Theory, Bagger, Lambert, Gutavsson, Lie,
    Filippov, 3-algebra},
  pdfauthor  = {Paul de Medeiros, José Figueroa-O'Farrill, Elena
    Méndez-Escobar, Patricia Ritter},
  pdfcreator = {\LaTeX\ with package \flqq hyperref\frqq}
}
\PrerenderUnicode{éÉ}
\newcommand{\half}{\tfrac12}
\newcommand{\fg}{\mathfrak{g}}

\newcommand{\fa}{\mathfrak{a}}

\newcommand{\fd}{\mathfrak{d}}

\newcommand{\fh}{\mathfrak{h}}

\newcommand{\fk}{\mathfrak{k}}
\newcommand{\fl}{\mathfrak{l}}

\newcommand{\fr}{\mathfrak{r}}
\newcommand{\fs}{\mathfrak{s}}
\newcommand{\ft}{\mathfrak{t}}

\newcommand{\fz}{\mathfrak{z}}

\newcommand{\fso}{\mathfrak{so}}

\newcommand{\fsu}{\mathfrak{su}}

\newcommand{\RR}{\mathbb{R}}
\newcommand{\CC}{\mathbb{C}}
\newcommand{\ZZ}{\mathbb{Z}}
\newcommand{\eL}{\mathscr{L}}

\newcommand{\eV}{\mathscr{V}}
\newcommand{\eD}{\mathscr{D}}
\newcommand{\eA}{\mathscr{A}}
\newcommand{\eF}{\mathscr{F}}
\newcommand{\be}{\boldsymbol{e}}

\DeclareMathOperator{\End}{End}

\DeclareMathOperator{\ad}{ad}

\theoremstyle{plain}
\newtheorem{lemma}{Lemma}

\newtheorem{theorem}[lemma]{Theorem}

\theoremstyle{definition}

\newcommand{\MUNCH}[1]{\relax}

\allowdisplaybreaks
\begin{document}
\title{Metric 3-Lie algebras for unitary Bagger--Lambert theories}
\author[de Medeiros, Figueroa-O'Farrill, Méndez-Escobar, Ritter]{Paul de
  Medeiros, José Figueroa-O'Farrill, Elena Méndez-Escobar and Patricia Ritter}
\address{Maxwell Institute and School of Mathematics,
  University of Edinburgh, UK}
\email{\{P.deMedeiros,J.M.Figueroa,E.Mendez\}@ed.ac.uk, P.D.Ritter@sms.ed.ac.uk}
\date{\today}
\begin{abstract}
  We prove a structure theorem for finite-dimensional indefinite-signature metric 3-Lie algebras admitting a maximally isotropic
  centre.  This algebraic condition indicates that all the negative-norm states in the associated Bagger--Lambert theory can be
  consistently decoupled from the physical Hilbert space.  As an immediate application of the theorem, new examples beyond index
  $2$ are constructed.  The lagrangian for the Bagger--Lambert theory based on a general physically admissible 3-Lie algebra of
  this kind is obtained.  Following an expansion around a suitable vacuum, the precise relationship between such theories and
  certain more conventional maximally supersymmetric gauge theories is found.  These typically involve particular combinations
  of $N=8$ super Yang-Mills and massive vector supermultiplets.  A dictionary between the 3-Lie algebraic data and the physical
  parameters in the resulting gauge theories will thereby be provided.
\end{abstract}
\maketitle
\tableofcontents

\section{Introduction and Summary}
\label{sec:introduction}

The fundamental ingredient in the Bagger--Lambert--Gustavsson (BLG) model \cite{BL1,GustavssonAlgM2,BL2}, proposed as the
low-energy effective field theory on a stack of coincident M2-branes, is a metric 3-Lie algebra $V$ on which the matter fields
take values.  This means that $V$ is a real vector space with a symmetric inner product $\left<-,-\right>$ and a trilinear,
alternating 3-bracket $[-,-,-]: V \times V \times V \to V$ obeying the fundamental identity \cite{Filippov}
\begin{equation}
  \label{eq:FI}
  [x,y,[z_1,z_2,z_3]] = [[x,y,z_1],z_2,z_3] + [z_1,[x,y,z_2],z_3] + [z_1,z_2,[x,y,z_3]]~,
\end{equation}
and the metricity condition
\begin{equation}
  \label{eq:metricity}
  \left<[x,y,z_1],z_2\right> = - \left<z_1,[x,y,z_2]\right>~,
\end{equation}
for all $x,y,z_i \in V$.  We say that $V$ is indecomposable if it is not isomorphic to an orthogonal direct sum of nontrivial
metric 3-Lie algebras.  Every indecomposable metric 3-Lie algebra gives rise to a BLG model and this motivates their
classification.  It is natural to attempt this classification in increasing index --- the index of an inner product being the
dimension of the maximum negative-definite subspace.  In other words, index $0$ inner products are positive-definite (called
euclidean here), index $1$ are lorentzian, et cetera.  To this date there is a classification up to index $2$, which we now
review.

It was conjectured in \cite{FOPPluecker} and proved in \cite{NagykLie} (see also \cite{GP3Lie,GG3Lie}) that there exists a
unique nonabelian indecomposable metric 3-Lie algebra of index $0$.  It is the simple 3-Lie algebra \cite{Filippov} $S_4$ with
underlying vector space $\RR^4$, orthonormal basis $e_1,e_2,e_3,e_4$, and 3-bracket
\begin{equation}
  \label{eq:S43b}
  [e_i,e_j,e_k] = \sum_{\ell =1}^4 \varepsilon_{ijk\ell} e_\ell~,
\end{equation}
where $\varepsilon = e_1 \wedge e_2 \wedge e_3 \wedge e_4$.  Nonabelian indecomposable 3-Lie algebras of index $1$ were
classified in \cite{Lor3Lie} and are given either by
\begin{itemize}
\item the simple lorentzian 3-Lie algebra $S_{3,1}$ with underlying vector space $\RR^4$, orthonormal basis $e_0,e_1,e_2,e_3$
  with $e_0$ timelike, and 3-bracket
  \begin{equation}
    \label{eq:S313b}
    [e_\mu,e_\nu,e_\rho] = \sum_{\sigma =0}^3 \varepsilon_{\mu\nu\rho\sigma} s_\sigma e_\sigma~,
  \end{equation}
  where $s_0 = -1$ and $s_i = 1$ for $i=1,2,3$; or
\item $W(\fg)$, with underlying vector space $\fg \oplus \RR u \oplus \RR v$, where $\fg$ is a semisimple Lie algebra with a
  choice of positive-definite invariant inner product, extended to $W(\fg)$ by declaring $u,v \perp \fg$ and $\left<u,u\right> =
  \left<v,v\right> = 0$ and $\left<u,v\right>=1$, and with 3-brackets
  \begin{equation}
    \label{eq:Wg3b}
    [u,x,y] = [x,y] \qquad\text{and}\qquad [x,y,z] = - \left<[x,y],z\right> v~,
  \end{equation}
  for all $x,y,z \in \fg$.
\end{itemize}
The latter metric 3-Lie algebras were discovered independently in \cite{GMRBL,BRGTV,HIM-M2toD2rev} in the context of the BLG
model.  The index 2 classification is presented in \cite{2p3Lie}.  There we found two classes of solutions, termed Ia and IIIb.
The former class is of the form $W(\fg)$, but where $\fg$ is now a lorentzian semisimple Lie algebra, whereas the latter class
will be recovered as a special case of the results in this paper and hence will be described in more detail below.

Let us now discuss the BLG model from a 3-algebraic perspective.  The $V$-valued matter fields in the BLG model
\cite{BL1,GustavssonAlgM2,BL2} comprise eight bosonic scalars $X$ and eight fermionic Majorana spinors $\Psi$ in three-dimensional
Minkowski space $\RR^{1,2}$.  Triality allows one to take the scalars $X$ and fermions $\Psi$ to transform respectively in the
vector and chiral spinor representations of the $\fso(8)$ R-symmetry.  These matter fields are coupled to a nondynamical gauge
field $A$ which is valued in $\Lambda^2V$ and described by a so-called twisted Chern--Simons term in the Bagger--Lambert
lagrangian \cite{BL1,BL2}.  The inner product $\left< -,- \right>$ on $V$ is used to describe the kinetic terms for the matter
fields $X$ and $\Psi$ in the Bagger--Lambert lagrangian.  Therefore if the index of $V$ is positive (i.e. not euclidean signature)
then the associated BLG model is not unitary as a quantum field theory, having `wrong' signs for the kinetic terms for those
matter fields in the negative-definite directions on $V$, thus carrying negative energy.

Indeed, for the BLG model based on the index-1 3-Lie algebra $W(\fg)$, one encounters just this problem.  Remarkably though, as
noted in the pioneering works \cite{GMRBL,BRGTV,HIM-M2toD2rev}, here the matter field components $X^v$ and $\Psi^v$ along
precisely one of the two null directions $(u,v)$ in $W(\fg)$ never appear in any of the interaction terms in the Bagger--Lambert
lagrangian.  Since the interactions are governed only by the structure constants of the 3-Lie algebra then this property simply
follows from the absence of $v$ on the left hand side of any of the 3-brackets in \eqref{eq:Wg3b}.  Indeed the one null direction
$v$ spans the centre of $W(\fg)$ and the linear equations of motion for the matter fields along $v$ force the components $X^u$ and
$\Psi^u$ in the other null direction $u$ to take constant values (preservation of maximal supersymmetry in fact requires $\Psi^u =
0$).  By expanding around this maximally supersymmetric and gauge-invariant vacuum defined by the constant expectation value of
$X^u$, one can obtain a unitary quantum field theory.  Use of this strategy in \cite{HIM-M2toD2rev} gave the first indication that
the resulting theory is nothing but $N=8$ super Yang--Mills theory on $\RR^{1,2}$ with the euclidean semi-simple gauge algebra
$\fg$.  The super Yang--Mills theory gauge coupling here being identified with the $SO(8)$-norm of the constant $X^u$.  This
procedure is somewhat reminiscent of the novel Higgs mechanism introduced in \cite{MukhiBL} in the context of the Bagger--Lambert
theory based on the euclidean Lie 3-algebra $S_4$.  In that case an $N=8$ super Yang-Mills theory with $\fsu(2)$ gauge algebra is
obtained, but with an infinite set of higher order corrections suppressed by inverse powers of the gauge coupling.  As found in
\cite{HIM-M2toD2rev}, the crucial difference is that there are no such corrections present in the lorentzian case.

Of course, one must be wary of naively integrating out the free matter fields $X^v$ and $\Psi^v$ in this way since their absence
in any interaction terms in the Bagger--Lambert lagrangian gives rise to an enhanced global symmetry that is generated by
shifting them by constant values.  To account for this degeneracy in the action functional, in order to correctly evaluate the
partition function, one must gauge the shift symmetry and perform a BRST quantisation of the resulting theory.  Fixing this
gauged shift symmetry allows one to set $X^v$ and $\Psi^v$ equal to zero while the equations of motion for the new gauge fields
sets $X^u$ constant and $\Psi^u = 0$.  Indeed this more rigorous treatment has been carried out in \cite{BLSNoGhost,GomisSCFT}
whereby the perturbative equivalence between the Bagger--Lambert theory based on $W( \fg )$ and maximally supersymmetric
Yang--Mills theory with euclidean gauge algebra $\fg$ was established (see also \cite{D2toD2}).  Thus the introduction of
manifest unitarity in the quantum field theory has come at the expense of realising an explicit maximal superconformal symmetry
in the BLG model for $W( \fg )$, i.e. scale-invariance is broken by a nonzero vacuum expectation value for $X^u$.  It is
perhaps worth pointing out that the super Yang--Mills description seems to have not captured the intricate structure of a
particular \lq degenerate' branch of the classical maximally supersymmetric moduli space in the BLG model for $W( \fg )$ found
in \cite{Lor3Lie}.  The occurrence of this branch can be understood to arise from a degenerate limit of the theory wherein the
scale $X^u = 0$ and maximal superconformal symmetry is restored.  However, as found in \cite{BLSNoGhost,GomisSCFT}, the
maximally superconformal unitary theory obtained by expanding around $X^u = 0$ describes a rather trivial free theory for eight
scalars and fermions, whose moduli space does not describe said degenerate branch of the original moduli space.

Consider now a general indecomposable metric 3-Lie algebra with index $r$ of the form $V = \bigoplus_{i=1}^r ( \RR u_i \oplus
\RR v_i ) \oplus W$, where $\left< u_i , u_j \right> = 0 = \left< v_i , v_j \right>$, $\left< u_i , v_j \right> = \delta_{ij}$
and $W$ is a euclidean vector space.  As explained in section 2.4 of \cite{2p3Lie}, one can ensure that none of the null
components $X^{v_i}$ and $\Psi^{v_i}$ of the matter fields appear in any of the interactions in the associated Bagger--Lambert
lagrangian provided that no $v_i$ appear on the left hand side of any of the 3-brackets on $V$.  This guarantees one has an
extra shift symmetry for each of these null components suggesting that all the associated negative-norm states in the spectrum
of this theory can be consistently decoupled after gauging all the shift symmetries and following BRST quantisation of the
gauged theory.  A more invariant way of stating the aforementioned criterion is that $V$ should admit a \emph{maximally
  isotropic centre}: that is, a subspace $Z\subset V$ of dimension equal to the index of the inner product on $V$, on which the
inner product vanishes identically and which is central, so that $[Z,V,V] = 0$ in the obvious notation.  The null directions
$v_i$ defined above along which we require the extra shift symmetries are thus taken to provide a basis for $Z$.  In
\cite{2p3Lie} we classified all indecomposable metric 3-Lie algebras of index $2$ with a maximally isotropic centre.  There are
nine families of such 3-Lie algebras, which were termed type IIIb in that paper.  In the present paper we will prove a structure
theorem for general metric 3-Lie algebras which admit a maximally isotropic centre, thus characterising them fully.  Although
the structure theorem falls short of a classification, we will argue that it is the best possible result for this problem.  The
bosonic contributions to the Bagger--Lambert lagrangians for such 3-Lie algebras will be computed but we will not perform a
rigorous analysis of the physical theory in the sense of gauging the shift symmetries and BRST quantisation.  We will limit
ourselves to expanding the theory around a suitable maximally supersymmetric and gauge-invariant vacuum defined by a constant
expectation value for $X^{u_i}$ (with $\Psi^{u_i} = 0$).  This is the obvious generalisation of the procedure used in
\cite{HIM-M2toD2rev} for the lorentzian theory and coincides with that used more recently in \cite{Ho:2009nk} for more general
3-Lie algebras.  We will comment explicitly on how all the finite-dimensional examples considered in section 4 of
\cite{Ho:2009nk} can be recovered from our formalism.

As explained in sections 2.5 and 2.6 of \cite{2p3Lie}, two more algebraic conditions are necessary in order to interpret the BLG
model based on a general metric 3-Lie algebra with maximally isotropic centre as an M2-brane effective field theory.  Firstly,
the 3-Lie algebra should admit a (nonisometric) conformal automorphism that can be used to absorb the formal coupling
dependence in the BLG model.  In \cite{2p3Lie} we determined that precisely four of the nine IIIb families of index $2$ 3-Lie
algebras with maximally isotropic centre satisfy this condition.  Secondly, parity invariance of the BLG model requires the
3-Lie algebra to admit an isometric antiautomorphism.  This symmetry is expected of an M2-brane effective field theory based on
the assumption that it should arise as an IR superconformal fixed point of $N=8$ super Yang--Mills theory.  In \cite{2p3Lie} we
determined that each of the four IIIb families of index $2$ 3-Lie algebras admitting said conformal automorphism also admitted
an isometric antiautomorphism.

It is worth emphasising that the motivation for the two conditions above is distinct from that which led us to demand a
maximally isotropic centre.  The first two are required only for an M-theoretic interpretation while the latter is a basic
physical consistency condition to ensure that the resulting quantum field theory is unitary.  Moreover, even given a BLG model
based on a 3-Lie algebra satisfying all three of these conditions, it is plain to see that the procedure we shall follow must
generically break the initial conformal symmetry since it has introduced scales into the problem corresponding to the vacuum
expectation values of $X^{u_i}$.  It is inevitable that this breaking of scale-invariance will also be a feature resulting from
a more rigorous treatment in terms of gauging shift symmetries and BRST quantisation.

Thus we shall concentrate just on the unitarity condition and, for the purposes of this paper, we will say that a metric 3-Lie
algebra is \textbf{(physically) admissible} if it is indecomposable and admits a maximally isotropic centre.  The first part of
the present paper will be devoted in essence to characterising finite-dimensional admissible 3-Lie algebras.  The second part
will describe the general structure of the gauge theories which result from expanding the BLG model based on these physically
admissible 3-Lie algebras around a given vacuum expectation value for $X^{u_i}$.  Particular attention will be paid to
explaining how the 3-Lie algebraic data translates into physical parameters of the resulting gauge theories.

This paper is organised as follows.  Section \ref{sec:towards-class-admiss} is concerned with the proof of Theorem \ref{thm:main},
which is outlined at the start of that section.  The theorem may be paraphrased as stating that every finite-dimensional
admissible 3-Lie algebra of index $r>0$ is constructed as follows.  We start with the following data:
\begin{itemize}
\item for each $\alpha=1,\dots,N$, a nonzero vector $0 \neq \kappa^\alpha \in \RR^r$ with components $\kappa^\alpha_i$, a
  positive real number $\lambda_\alpha >0$ and a compact simple Lie algebra $\fg_\alpha$;

\item for each $\pi =1,\dots, M$, a two-dimensional euclidean vector space $E_\pi$ with a complex structure $H_\pi$, and two
  linearly independent vectors $\eta^\pi, \zeta^\pi \in \RR^r$;

\item a euclidean vector space $E_0$ and $K \in \Lambda^3\RR^r \otimes E_0$ obeying the quadratic equations
  \begin{equation*}
    \left<K_{ijn},K_{k\ell m}\right> - \left<K_{ijm},K_{nk\ell}\right> + \left<K_{ij\ell},K_{mnk}\right> - \left<K_{ijk},K_{\ell mn}\right>  = 0,
  \end{equation*}
  where $\left<-,-\right>$ is the inner product on $E_0$;

\item and $L \in \Lambda^4 \RR^r$.
\end{itemize}
On the vector space
\begin{equation*}
  V = \bigoplus_{i=1}^r \left(\RR u_i \oplus \RR v_i\right) \oplus \bigoplus_{\alpha =1}^N \fg_\alpha \oplus \bigoplus_{\pi=1}^M
  E_\pi \oplus E_0,
\end{equation*}
we define the following inner product extending the inner product on
$E_\pi$ and $E_0$:
\begin{itemize}
\item $\left<u_i,v_j\right>=\delta_{ij}$, $\left<u_i,u_j\right>=0$, $\left<v_i,v_j\right> = 0$ and $u_i,v_j$ are orthogonal to
  the $\fg_\alpha$, $E_\pi$ and $E_0$; and
\item on each $\fg_\alpha$ we take $-\lambda_\alpha$ times the Killing form.
\end{itemize}
This makes $V$ above into an inner product space of index $r$.  On $V$ we define the following 3-brackets, with the tacit
assumption that any 3-bracket not listed here is meant to vanish:
\begin{equation}
  \begin{aligned}[m]
    [u_i,u_j,u_k] &= K_{ijk} + \sum_{\ell=1}^r L_{ijk\ell} v_\ell\\
    [u_i,u_j,x_0] &= - \sum_{k=1}^r \left<K_{ijk},x_0\right> v_k\\
    [u_i,u_j,x_\pi] &= (\eta^\pi_i \zeta^\pi_j - \eta^\pi_j \zeta^\pi_i) H_\pi x_\pi\\
    [u_i,x_\pi,y_\pi] &= \left<H_\pi x_\pi,y_\pi\right> \sum_{j=1}^r (\eta^\pi_i \zeta^\pi_j -\eta^\pi_j \zeta^\pi_i) v_j\\
    [u_i,x_\alpha,y_\alpha] &= \kappa_i^\alpha [x_\alpha,y_\alpha]\\
    [x_\alpha,y_\alpha,z_\alpha] &= - \left<[x_\alpha,y_\alpha],z_\alpha\right> \sum_{i=1}^r \kappa_i^\alpha v_i,
  \end{aligned}
\end{equation}
for all $x_0 \in E_0$, $x_\pi,y_\pi\in E_\pi$, and $x_\alpha, y_\alpha, z_\alpha \in \fg_\alpha$.  The resulting metric 3-Lie
algebra has a maximally isotropic centre spanned by the $v_i$.  It is indecomposable provided that there is no $x_0 \in E_0$ which
is perpendicular to all the $K_{ijk}$, whence in particular $\dim E_0 \leq \binom{r}{3}$.  The only non-explicit datum in the
above construction are the $K_{ijk}$ since they are subject to certain quadratic equations.  However we will see that these
equations are trivially satisfied for $r<5$.  Hence the above results constitutes, in principle, a classification for indices 3
and 4, extending the classification of index 2 in \cite{2p3Lie}.

Using this structure theorem we are able to calculate the lagrangian for the BLG model associated with a general physically
admissible 3-Lie algebra.  For the sake of clarity, we shall focus on just the bosonic contributions since the resulting theories
will have a canonical maximally supersymmetric completion.  Upon expanding this theory around the maximally supersymmetric vacuum
defined by constant expectation values $X^{u_i}$ (with all the other fields set to zero) we will obtain standard $N=8$
supersymmetric (but nonconformal) gauge theories with moduli parametrised by particular combinations of the data appearing in
Theorem \ref{thm:main} and the vacuum expectation values $X^{u_i}$.  It will be useful to think of the vacuum expectation values
$X^{u_i}$ as defining a linear map, also denoted $X^{u_i}: \RR^r \to \RR^8$, sending $\xi \mapsto X^\xi := \sum_{i=1}^r \xi_i
X^{u_i}$.  Indeed it will be found that the physical gauge theory parameters are naturally expressed in terms of components in the
image of this map.  That is, in general, we find that neither the data in Theorem \ref{thm:main} nor the vacuum expectation values
$X^{u_i}$ on their own appear as physical parameters which instead arise from certain projections of the components of the data in
Theorem \ref{thm:main} onto $X^{u_i}$ in $\RR^8$.

The resulting Bagger--Lambert lagrangian will be found to factorise into a sum of decoupled maximally supersymmetric gauge
theories on each of the euclidean components $\fg_\alpha$, $E_\pi$ and $E_0$.  The physical content and moduli on each component
can be summarised as follows:
\begin{itemize}
\item On each $\fg_\alpha$ one has an $N=8$ super Yang--Mills theory.  The gauge symmetry is based on the simple Lie algebra
  $\fg_\alpha$.  The coupling constant is given by $\| X^{\kappa^\alpha} \|$, which denotes the $SO(8)$-norm of the image of
  $\kappa^\alpha \in \RR^r$ under the linear map $X^{u_i}$.  The seven scalar fields take values in the hyperplane $\RR^7 \subset
  \RR^8$ which is orthogonal to the direction defined by $X^{\kappa^\alpha}$. (If $X^{\kappa^\alpha}=0$, for a given value of
  $\alpha$, one obtains a degenerate limit corresponding to a maximally superconformal free theory for eight scalar fields and
  eight fermions valued in $\fg_\alpha$.)

\item On each plane $E_\pi$ one has a pair of identical free abelian $N=8$ massive vector supermultiplets.  The bosonic fields in
  each such supermultiplet comprise a massive vector and six massive scalars.  The mass parameter is given by $\| X^{\eta^\pi}
  \wedge X^{\zeta^\pi} \|$, which corresponds to the area of the parallelogram in $\RR^8$ defined by the vectors $X^{\eta^\pi}$
  and $X^{\zeta^\pi}$ in the image of the map $X^{u_i}$.  The six scalar fields inhabit the $\RR^6 \subset \RR^8$ which is
  orthogonal to the plane spanned by $X^{\eta^\pi}$ and $X^{\zeta^\pi}$. (If $\| X^{\eta^\pi} \wedge X^{\zeta^\pi} \| =0$, for a
  given value of $\pi$, one obtains a degenerate massless limit where the vector is dualised to a scalar, again corresponding to a
  maximally superconformal free theory for eight scalar fields and eight fermions valued in $E_\pi$.) Before gauge-fixing, this
  theory can be understood as an $N=8$ super Yang--Mills theory with gauge symmetry based on the four-dimensional Nappi--Witten
  Lie algebra $\fd(E_\pi,\RR)$.  Moreover we explain how it can be obtained from a particular truncation of an $N=8$ super
  Yang-Mills theory with gauge symmetry based on any euclidean semisimple Lie algebra with rank 2, which may provide a more
  natural D-brane interpretation.

\item On $E_0$ one has a decoupled $N=8$ supersymmetric theory involving eight free scalar fields and an abelian Chern--Simons
  term.  Since none of the matter fields are charged under the gauge field in this Chern--Simons term then its overall
  contribution is essentially trivial on $\RR^{1,2}$.
\end{itemize}

\section*{Note added}

During the completion of this work the paper \cite{Ho:2009nk} appeared whose results have noticeable overlap with those found
here.  In particular, they also describe the physical properties of BLG models based on certain finite-dimensional 3-Lie algebras
with index greater than 1 admitting a maximally isotropic centre.  The structure theorem we prove here for such 3-Lie algebras
allows us to extend some of their results and make general conclusions about the nature of those unitary gauge theories which
arise from BLG models based on physically admissible 3-Lie algebras.  In terms of our data in Theorem \ref{thm:main}, the explicit
finite-dimensional examples considered in section 4 of \cite{Ho:2009nk} all have $K_{ijk} = 0 = L_{ijkl}$ with only one $J_{ij}$
nonzero.  This is tantamount to taking the index $r=2$.  The example in sections 4.1 and 4.2 of \cite{Ho:2009nk} has
$\kappa^\alpha = 0$ (i.e. no $\fg_\alpha$ part) while the example in section 4.3 has $\kappa^\alpha = (1,0)^t$.  These are
isomorphic to two of the four physically admissible IIIb families of index $2$ 3-Lie algebras found in \cite{2p3Lie}.

\section{Towards a classification of admissible metric 3-Lie algebras}
\label{sec:towards-class-admiss}

In this section we will prove a structure theorem for
finite-dimensional indecomposable metric 3-Lie algebras admitting a
maximally isotropic centre.  We think it is of pedagogical value to
first rederive the similar structure theorem for metric Lie algebras
using a method similar to the one we will employ in the more involved
case of metric 3-Lie algebras.

\subsection{Metric Lie algebras with maximally isotropic centre}
\label{sec:metric-lie-algebras}

Recall that a Lie algebra $\fg$ is said to be metric, if it possesses an ad-invariant scalar product.  It is said to be
indecomposable if it is not isomorphic to an orthogonal direct sum of metric Lie algebras (of positive dimension).  Equivalently,
it is indecomposable if there are no proper ideals on which the scalar product restricts nondegenerately.  A metric Lie algebra
$\fg$ is said to have index $r$, if the ad-invariant scalar product has index $r$, which is the same as saying that the maximally
negative-definite subspace of $\fg$ is $r$-dimensional.  In this section we will prove a structure theorem for finite-dimensional
indecomposable metric Lie algebras admitting a maximally isotropic centre, a result originally due to Kath and Olbrich
\cite{KathOlbrich2p}.

\subsubsection{Preliminary form of the Lie algebra}
\label{sec:preliminaries}

Let $\fg$ be a finite-dimensional indecomposable metric Lie algebra
of index $r>0$ admitting a maximally isotropic centre.   Let $v_i$,
$i=1,\dots,r$, denote a basis for the centre.  The inner product is
such that $\left<v_i,v_j\right>=0$.  Since the inner product on $\fg$
is nondegenerate, there exist $u_i$, $i=1,\dots,r$, which obey
$\left<u_i,v_j\right> = \delta_{ij}$.  It is always possible to
choose the $u_i$ such that $\left<u_i,u_j\right>=0$.  Indeed, if the
$u_i$ do not span a maximally isotropic subspace, then redefine them
by $u_i \mapsto u_i - \half \sum_{j=1}^r \left<u_i,u_j\right> v_j$ so
that they do.  The perpendicular complement to the $2r$-dimensional
subspace spanned by the $u_i$ and the $v_j$ is then
positive-definite.  In summary, $\fg$ admits the following vector
space decomposition
\begin{equation}
  \label{eq:decomp-2Lie}
  \fg = \bigoplus_{i=1}^r \left(\RR u_i \oplus \RR v_i\right) \oplus
  \fr,
\end{equation}
where $\fr$ is the positive-definite subspace of $\fg$ perpendicular
to all the $u_i$ and $v_j$.

Metricity then implies that the most general Lie brackets on $\fg$ are
of the form
\begin{equation}
  \label{eq:2-brackets}
  \begin{aligned}[m]
    [u_i,u_j] &= K_{ij} + \sum_{k=1}^r L_{ijk} v_k\\
    [u_i,x] &=  J_ix - \sum_{j=1}^r \left<K_{ij},x\right> v_j\\
    [x,y] &= [x,y]_{\fr} - \sum_{i=1}^r \left<x, J_i y\right> v_i,
  \end{aligned}
\end{equation}
where $K_{ij} = - K_{ji} \in \fr$, $L_{ijk} \in \RR$ is totally
skewsymmetric in the indices, $J_i \in \fso(\fr)$ and $[-,-]_\fr:\fr
\times \fr \to \fr$ is bilinear and skewsymmetric.  Metricity and the
fact that the $v_i$ are central, means that no $u_i$ can appear on the
right-hand side of a bracket.  Finally, metricity also implies that
\begin{equation}
  \label{eq:metric-Lie}
   \left<[x,y]_\fr, z \right> = \left<x, [y,z]_\fr \right>,
\end{equation}
for all $x,y,z \in \fr$.

It is not hard to demonstrate that the Jacobi identity for $\fg$ is
equivalent to the following identities on $[-,-]_{\fr}$, $J_i$ and
$K_{ij}$, whereas $L_{ijk}$ is unconstrained:
\begin{subequations}
  \begin{align}
    [x,[y,z]_\fr]_\fr - [[x,y]_\fr,z]_\fr - [y,[x,z]_\fr]_\fr &= 0    \label{eq:h-is-Lie}\\
    J_i [x,y]_{\fr} - [J_ix,y]_{\fr} -  [x, J_i y]_{\fr} &=  0 \label{eq:J-is-deriv}\\
    J_i J_j x - J_j J_i x - [K_{ij},x]_\fr &= 0 \label{eq:JJK}\\
    J_i K_{jk} + J_j K_{ki} +     J_k K_{ij}  &= 0 \label{eq:JK}\\
    \left<K_{\ell i},K_{jk}\right> + \left<K_{\ell j},K_{ki}\right> +
    \left<K_{\ell k},K_{ij}\right> &=0, \label{eq:KK}
  \end{align}
\end{subequations}
for all $x,y,z \in \fr$.

\subsubsection{$\fr$ is abelian}
\label{sec:fr-abelian}

Equation \eqref{eq:h-is-Lie} says that $\fr$ is a Lie algebra under
$[-,-]_\fr$, which because of equation \eqref{eq:metric-Lie} is
metric.  Being positive-definite, it is reductive, whence an orthogonal
direct sum $\fr = \fs \oplus \fa$, where $\fs$ is semisimple and $\fa$
is abelian.  We will show that for an indecomposable $\fg$, we are
forced to take $\fs=0$, by  showing that $\fg = \fs \oplus \fs^\perp$
as a metric Lie algebra.

Equation \eqref{eq:J-is-deriv} says that $J_i$ is a derivation of
$\fr$, which we know to be skewsymmetric.  The Lie algebra of
skewsymmetric derivations of $\fr$ is given by $\ad \fs \oplus
\fso(\fa)$.  Therefore under this decomposition, we may write $J_i =
\ad z_i + J^\fa_i$, for some unique $z_i \in \fs$ and $J^\fa_i \in
\fso(\fa)$.

Decompose $K_{ij} = K^{\fs}_{ij} + K^{\fa}_{ij}$, with $K^{\fs}_{ij}
\in \fs$ and $K^{\fa}_{ij} \in \fa$.  Then equation \eqref{eq:JJK}
becomes the following two conditions
\begin{align}
  [z_i,z_j]_\fr &= K^{\fs}_{ij}   \label{eq:zzKs}\\
  \intertext{and}
  [J^\fa_i, J^\fa_j] &= 0.  \label{eq:Js-commute}
\end{align}

One can now check that the $\fs$-component of the Jacobi identity for
$\fg$ is automatically satisfied, whereas the $\fa$-component gives
rise to the two equations
\begin{align}
  J^\fa_i K^{\fa}_{jk} +   J^\fa_j K^{\fa}_{ki} +   J^\fa_k K^{\fa}_{ij} &= 0 \label{eq:JKa}\\
  \intertext{and}
    \left<K^{\fa}_{\ell i},K^{\fa}_{jk}\right> + \left<K^{\fa}_{\ell j},K^{\fa}_{ki}\right> +
    \left<K^{\fa}_{\ell k},K^{\fa}_{ij}\right> &=0. \label{eq:KKa}
\end{align}

We will now show that $\fg \cong \fs \oplus \fs^\perp$, which violates
the indecomposability of $\fg$ unless $\fs = 0$.  Consider the
isometry $\varphi$ of the vector space $\fg$ defined by
\begin{equation}
  \label{eq:isometry-2}
  \begin{aligned}[m]
    \varphi(u_i) &= u_i - z_i - \half \sum_{j=1}^r \left<z_i,z_j\right> v_j\\
    \varphi(v_i) &= v_i\\
    \varphi(x) &= x + \sum_{i=1}^r \left<z_i, x\right> v_i,
  \end{aligned}
\end{equation}
for all $x \in \fr$.  Notice that if $x \in \fa$, then $\varphi(x) =
x$.  It is a simple calculation to see that for all $x,y\in \fs$,
\begin{equation}
  [\varphi(u_i), \varphi(x)] = 0 \qquad\text{and}\qquad [\varphi(x),
  \varphi(y)] = \varphi ([x,y]_\fr).
\end{equation}
In other words, the image of $\fs$ under $\varphi$ is a Lie subalgebra
of $\fg$ isomorphic to $\fs$ and commuting with its perpendicular
complement in $\fg$.  In other words, as a metric Lie algebra $\fg
\cong \fs \oplus \fs^\perp$, violating the decomposability of $\fg$
unless $\fs = 0$.

In summary, we have proved the following

\begin{lemma}\label{lem:h-is-abelian}
  Let $\fg$ be a finite-dimensional indecomposable metric Lie algebra
  with index $r>0$ and admitting a maximally isotropic centre.  Then
  as a vector space
  \begin{equation}
    \fg = \bigoplus_{i=1}^r \left(\RR u_i \oplus \RR v_i\right) \oplus E,
  \end{equation}
  where $E$ is a euclidean space, $u_i, v_i \perp E$ and
  $\left<u_i,v_j\right>= \delta_{ij}$, $\left<u_i,u_j\right> =
  \left<v_i,v_j\right> = 0$.  Moreover the Lie bracket is given by
  \begin{equation}
    \label{eq:Lie-brackets}
    \begin{aligned}[m]
      [u_i,u_j] &= K_{ij} + \sum_{k=1}^r L_{ijk} v_k\\
      [u_i,x] &=  J_ix - \sum_{j=1}^r \left<K_{ij},x\right> v_j\\
      [x,y] &= - \sum_{i=1}^r \left<x, J_i y\right> v_i,
    \end{aligned}
  \end{equation}
  where $K_{ij} = - K_{ji} \in E$, $L_{ijk} \in \RR$ is totally
  skewsymmetric in its indices, $J_i \in \fso(E)$ and in addition obey
  the following conditions:
  \begin{subequations}
    \begin{align}
      J_i J_j - J_j J_i &= 0 \label{eq:Jscommute}\\
      J_i K_{jk} + J_j K_{ki} +     J_k K_{ij}  &= 0 \label{eq:JK2}\\
      \left<K_{\ell i},K_{jk}\right> + \left<K_{\ell j},K_{ki}\right> +
      \left<K_{\ell k},K_{ij}\right> &=0. \label{eq:KK2}
    \end{align}
  \end{subequations}
\end{lemma}
The analysis of the above equations will take the rest of this
section, until we arrive at the desired structure theorem.

\subsubsection{Solving for the $J_i$}
\label{sec:solving-js}

Equation \eqref{eq:Jscommute} says that the $J_i \in \fso(E)$ are
mutually commuting, whence they span an abelian subalgebra $\fh
\subset \fso(E)$.  Since $E$ is positive-definite, $E$ decomposes as the
following orthogonal direct sum as a representation of $\fh$:
\begin{equation}
  \label{eq:E-split}
  E = \bigoplus_{\pi=1}^s E_\pi \oplus E_0,
\end{equation}
where
\begin{equation}
  \label{eq:e0}
  E_0 = \left\{ x\in E \middle | J_i x = 0~\forall i\right\}
\end{equation}
and each $E_\pi$ is a two-dimensional real irreducible representation of
$\fh$ with certain nonzero weight.  Let $(H_\pi)$ denote the basis for
$\fh$ where
\begin{equation}
  \label{eq:Hs}
  H_\pi H_\varrho =
  \begin{cases}
    0 & \text{if $\pi\neq \varrho$,}\\
    - \Pi_\pi & \text{if $\pi = \varrho$,}
  \end{cases}
\end{equation}
where $\Pi_\pi \in \End(E)$ is the orthogonal projector onto $E_\pi$.
Relative to this basis we can then write $J_i =\sum_\pi J_i^\pi H_\pi$, for
some real numbers $J_i^\pi$.

\subsubsection{Solving for the $K_{ij}$}
\label{sec:solving-ks}

Since $K_{ij} \in E$, we may decompose according to \eqref{eq:E-split} as
\begin{equation}
  K_{ij} = \sum_{\pi=1}^s K^\pi_{ij} + K^0_{ij}.
\end{equation}
We may identify each $E_\pi$ with a complex line where $H_\pi$ acts by
multiplication by $i$.  This turns the complex number $K^\pi_{ij}$ into
one component of a complex bivector $K^\pi \in \Lambda^2\CC^r$.
Equation \eqref{eq:JK2} splits into one equation for each $K^\pi$ and
that equation says that
\begin{equation}
  J^\pi_i K^\pi_{jk} + J^\pi_j K^\pi_{ki} + J^\pi_k K^\pi_{ij} = 0,
\end{equation}
or equivalently that $J^\pi \wedge K^\pi = 0$, which has as unique
solution $K^\pi = J^\pi \wedge t^\pi$, for some $t^\pi \in \RR^r$.  In other
words,
\begin{equation}
  K^\pi_{ij} = J^\pi_i t^\pi_j - J^\pi_j t^\pi_i.
\end{equation}

Now consider the following vector space isometry $\varphi: \fg \to
\fg$, defined by
\begin{equation}
  \label{eq:isometry-3}
  \begin{aligned}[m]
    \varphi(u_i) &= u_i - t_i - \half \sum_{j=1}^r \left<t_i,t_j\right> v_j\\
    \varphi(v_i) &= v_i\\
    \varphi(x) &= x + \sum_{i=1}^r \left<t_i, x\right> v_i,
  \end{aligned}
\end{equation}
for all $x \in E$, where $t_i \in E$ and hence $t_i = \sum_{\pi=1}^s
t_i^\pi + t_i^0$.  Under this isometry the form of the Lie algebra
remains invariant, but $K_{ij}$ changes as
\begin{equation}
  K_{ij} \mapsto K_{ij} - J_i t_j + J_j t_i
\end{equation}
and $L_{ijk}$ changes in a manner which need not concern us here.
Therefore we see that $K^\pi_{ij}$ has been put to zero via this
transformation, whereas $K^0_{ij}$ remains unchanged.  In other words,
we can assume without loss of generality that $K_{ij} \in E_0$, so
that $J_i K_{kl} = 0$, while still being subject to the quadratic
equation \eqref{eq:KK2}.

In summary, we have proved the following theorem, originally due to
Kath and Olbrich \cite{KathOlbrich2p}:

\begin{theorem}\label{thm:main-Lie}
  Let $\fg$ be a finite-dimensional indecomposable metric Lie algebra
  of index $r>0$ admitting a maximally isotropic centre.  Then as a
  vector space
  \begin{equation}
    \fg = \bigoplus_{i=1}^r \left(\RR u_i \oplus \RR v_i\right) \oplus
    \bigoplus_{\pi=1}^s E_\pi \oplus E_0,
  \end{equation}
  where all direct sums but the one between $\RR u_i$ and $\RR v_i$
  are orthogonal and the inner product is as in Lemma
  \ref{lem:h-is-abelian}.  Let $0 \neq J^\pi \in \RR^r$, $K_{ij} \in
  E_0$ and $L_{ijk} \in \RR$ and assume that the $K_{ij}$ obey the
  following quadratic relation
  \begin{equation}
    \label{eq:KK3}
    \left<K_{\ell i},K_{jk}\right> + \left<K_{\ell j},K_{ki}\right> +
    \left<K_{\ell k},K_{ij}\right>. =0.
  \end{equation}
  Then the Lie bracket of $\fg$ is given by
  \begin{equation}
    \label{eq:Lie-brackets-final}
    \begin{aligned}[m]
      [u_i,u_j] &= K_{ij} + \sum_{k=1}^r L_{ijk} v_k\\
      [u_i,x] &=  J^\pi_iH_\pi x\\
      [u_i,z] &= - \sum_{j=1}^r \left<K_{ij},z\right> v_j\\
      [x,y] &= - \sum_{i=1}^r \left<x, J^\pi_i H_\pi y\right> v_i,
    \end{aligned}
  \end{equation}
  where $x,y \in E_\pi$ and $z \in E_0$.  Furthermore, indecomposability
  forces the $K_{ij}$ to span all of $E_0$, whence $\dim E_0 \leq \binom{r}{2}$.
\end{theorem}

It should be remarked that the $L_{ijk}$ are only defined up to the
following transformation
\begin{equation}
  L_{ijk} \mapsto L_{ijk} + \left<K_{ij},t_k\right> +
  \left<K_{ki},t_j\right> + \left<K_{jk},t_i\right>,
\end{equation}
for some $t_i \in E_0$.

It should also be remarked that the quadratic relation \eqref{eq:KK3}
is automatically satisfied for index $r\leq 3$, whereas for index
$r\geq 4$ it defines an algebraic variety.  In that sense, the
classification problem for indecomposable metric Lie algebras
admitting a maximally isotropic centre is not tame for index $r>3$.

\subsection{Metric 3-Lie algebras with maximally isotropic centre}
\label{sec:max-iso-cent}

After the above warm-up exercise, we may now tackle the problem of interest, namely the classification of finite-dimensional
indecomposable metric 3-Lie algebras with maximally isotropic centre.  The proof is not dissimilar to that of
Theorem~\ref{thm:main-Lie}, but somewhat more involved and requires new ideas.  Let us summarise the main steps in the proof.
\begin{enumerate}
\item In section~\ref{sec:preliminary-form-3} we write down the most general form of a metric 3-Lie algebra $V$ consistent with
  the existence of a maximally isotropic centre $Z$.  As a vector space, $V = Z \oplus Z^* \oplus W$, where $Z$ and $Z^*$ are
  nondegenerately paired and $W$ is positive-definite.  Because $Z$ is central, the 4-form
  $F(x,y,z,w):=\left<[x,y,z],w\right>$ on $V$ defines an element in $\Lambda^4(W \oplus Z)$.  The decomposition
  \begin{equation}
    \Lambda^4(W \oplus Z) = \Lambda^4W \oplus \left( \Lambda^3W \otimes Z \right) \oplus \left(\Lambda^2W \otimes
      \Lambda^2Z\right) \oplus \left(W \otimes \Lambda^3 Z \right) \oplus \Lambda^4Z
  \end{equation}
  induces a decomposition of $F = \sum_{a=0}^4 F_a$, where $F_a \in \Lambda^{4-a}W \otimes \Lambda^aZ$, where the
  component $F_4$ is unconstrained.

\item The component $F_0$ defines the structure of a metric 3-Lie algebra on $W$ which, if $V$ is indecomposable, must be
  abelian, as shown in section~\ref{sec:w-abelian}.

\item The component $F_1$ defines a compatible family $[-,-]_i$ of reductive Lie algebras on $W$.  In
  section~\ref{sec:solving-brackets} we show that they all are proportional to a reductive Lie algebra structure $\fg \oplus \fz$
  on $W$, where $\fg$ is semisimple and $\fz$ is abelian.

\item In section~\ref{sec:solving-Js} we show that the component $F_2$ defines a family $J_{ij}$ of commuting endomorphisms
  spanning an abelian Lie subalgebra $\fa < \fso(\fz)$.  Under the action of $\fa$, $\fz$ breaks up into a direct sum of
  irreducible 2-planes $E_\pi$ and a euclidean vector space $E_0$ on which the $J_{ij}$ act trivially.

\item In section~\ref{sec:solving-K} we show that the component $F_3$ defines elements $K_{ijk} \in E_0$ which are subject to
  a quadratic equation.
\end{enumerate}

\subsubsection{Preliminary form of the 3-algebra}
\label{sec:preliminary-form-3}

Let $V$ be a finite-dimensional metric 3-Lie algebra with index $r>0$
and admitting a maximally isotropic centre.  Let $v_i$, $i=1,\dots,r$,
denote a basis for the centre.  Since the centre is (maximally)
isotropic, $\left<v_i,v_j\right>=0$, and since the inner product on
$V$ is nondegenerate, there exists $u_i$, $i=1,\dots,r$ satisfying
$\left<u_i,v_j\right> = \delta_{ij}$.  Furthermore, it is possible to
choose the $u_i$ such that $\left<u_i,u_j\right>=0$.  The
perpendicular complement $W$ of the $2r$-dimensional subspace spanned
by the $u_i$ and $v_i$ is therefore positive definite.  In other
words, $V$ admits a vector space decomposition
\begin{equation}
  V = \bigoplus_{i=1}^r \left( \RR u_i \oplus \RR v_i\right) \oplus W.
\end{equation}
Since the $v_i$ are central, metricity of $V$ implies that the $u_i$
cannot appear in the right-hand side of any 3-bracket.  The most
general form for the 3-bracket for $V$ consistent with $V$ being a
metric 3-Lie algebra is given for all $x,y,z \in W$ by
\begin{equation}
  \label{eq:3-brackets}
  \begin{aligned}[m]
    [u_i,u_j,u_k] &= K_{ijk} + \sum_{\ell=1}^r L_{ijk\ell} v_\ell\\
    [u_i,u_j,x] &= J_{ij} x - \sum_{k=1}^r \left<K_{ijk},x\right> v_k\\
    [u_i,x,y] &= [x,y]_i - \sum_{j=1}^r \left<x,J_{ij} y\right> v_j\\
    [x,y,z] &= [x,y,z]_W - \sum_{i=1}^r \left<[x,y]_i,z\right> v_i,
  \end{aligned}
\end{equation}
where $J_{ij} \in \fso(W)$, $K_{ijk} \in W$ and $L_{ijk\ell} \in \RR$
are skewsymmetric in their indices, $[-,-]_i : W \times W \to W$ is an
alternating bilinear map which in addition obeys
\begin{equation}
  \label{eq:metricity-i}
  \left<[x,y]_i,z\right> = \left<x,[y,z]_i\right>,
\end{equation}
and $[-,-,-]_W : W \times W \times W \to W$ is an alternating
trilinear map which obeys
\begin{equation}
  \label{eq:metricity-W}
  \left<[x,y,z]_W,w\right> = - \left<[x,y,w]_W,z\right>.
\end{equation}

The following lemma is the result of a straightforward, if somewhat
lengthy, calculation.

\begin{lemma}
  The fundamental identity \eqref{eq:FI} of the 3-Lie algebra $V$
  defined by \eqref{eq:3-brackets} is equivalent to the following
  conditions, for all $t,w,x,y,z\in W$:
  \begin{subequations}\label{eq:FI-V-pre}
    \begin{align}
      [t,w,[x,y,z]_W]_W &= [[t,w,x]_W,y,z]_W +[x,[t,w,y]_W,z]_W + [x,y,[t,w,z]_W]_W \label{eq:W-3la}\\
      [w,[x,y,z]_W]_i &= [[w,x]_i,y,z]_W + [x,[w,y]_i,z]_W + [x,y,[w,z]_i]_W \label{eq:ad-i-der-W}\\
      [x,y,[z,t]_i]_W &= [z,t,[x,y]_i]_W + [[x,y,z]_W,t]_i + [z,[x,y,t]_W]_i \label{eq:C9-W}\\
      J_{ij}[x,y,z]_W &=  [J_{ij}x,y,z]_W + [x,J_{ij}y,z]_W + [x,y,J_{ij}z]_W \label{eq:J-der-W}\\
      J_{ij}[x,y,z]_W - [x,y,J_{ij}z]_W  &= [[x,y]_i,z]_j - [[x,y]_j,z]_i \label{eq:C3-W}\\
      [x,y,K_{ijk}]_W &= J_{jk} [x,y]_i + J_{ki}[x,y]_j + J_{ij}[x,y]_k \label{eq:C6-W}\\
      [J_{ij}x,y,z]_W &= [[x,y]_i,z]_j + [[y,z]_j,x]_i + [[z,x]_i,y]_j \label{eq:C2-W}\\
      J_{ij}[x,y,z]_W &= [z,[x,y]_j]_i + [x,[y,z]_j]_i + [y,[z,x]_j]_i \label{eq:C2-W-bis}\\
      [x,y,K_{ijk}]_W &= J_{ij} [x,y]_k - [J_{ij} x,y ]_k - [x, J_{ij} y ]_k \label{eq:C1-W}\\
      J_{ik}[x,y]_j - J_{ij}[x,y]_k &= [J_{jk}x,y]_i + [x,J_{jk}y]_i \label{eq:J-der-2-1}\\
      [x,J_{jk}y]_i &= [J_{ij}x,y]_k + [J_{ki}x,y]_j + J_{jk}[x,y]_i \label{eq:J-der-2-2}\\
      [K_{ijk},x]_\ell &= [K_{\ell ij},x]_k + [K_{\ell jk},x]_i + [K_{\ell ki},x]_j \label{eq:K-ad-x}\\
      [K_{ijk},x]_\ell - [K_{ij\ell},x]_k &= \left(J_{ij}J_{k\ell} - J_{k\ell}J_{ij}\right) x  \label{eq:C4-W}\\
      [x,K_{jk\ell}]_i &= \left(J_{jk}J_{i\ell} + J_{k\ell}J_{ij} + J_{j\ell}J_{ki}\right) x  \label{eq:C5-W}\\
      J_{im}K_{jk\ell} &= J_{ij} K_{k\ell m} + J_{ik} K_{\ell mj} + J_{i\ell} K_{jkm} \label{eq:C7-bis}\\
      J_{ij}K_{k\ell m} &= J_{\ell m}K_{ijk}  + J_{mk}K_{ij\ell} + J_{k\ell}K_{ijm}  \label{eq:C7-W}\\
      \left<K_{ijm},K_{nk\ell}\right> + \left<K_{ijk},K_{\ell mn}\right>  &= \left<K_{ijn},K_{k\ell m}\right> + \left<K_{ij\ell},K_{mnk}\right> . \label{eq:C8-W}
    \end{align}
  \end{subequations}
\end{lemma}
Of course, not all of these equations are independent, but we will not
attempt to select a minimal set here, since we will be able to
dispense with some of the equations easily.

\subsubsection{$W$ is abelian}
\label{sec:w-abelian}

Equation \eqref{eq:W-3la} says that $W$ becomes a 3-Lie algebra under
$[-,-,-]_W$ which is metric by \eqref{eq:metricity-W}.  Since $W$ is
positive-definite, it is reductive
\cite{NagykLie,GP3Lie,GG3Lie,Lor3Lie}, whence isomorphic to an
orthogonal direct sum $W = S \oplus A$, where $S$ is semisimple and
$A$ is abelian.  Furthermore, $S$ is an orthogonal direct sum of
several copies of the unique positive-definite simple 3-Lie algebra
$S_4$ \cite{Filippov,LingSimple}.  We will show that as metric 3-Lie
algebras $V = S \oplus S^\perp$, whence if $V$ is indecomposable then
$S=0$ and $W=A$ is abelian as a 3-Lie algebra.  This is an extension
of the result in \cite{Lor3Lie} by which semisimple 3-Lie algebras $S$
factorise out of one-dimensional double extensions, and we will, in
fact, follow a similar method to the one in \cite{Lor3Lie} by which we
perform an isometry on $V$ which manifestly exhibits a nondegenerate
ideal isomorphic to $S$ as a 3-Lie algebra.

Consider then the isometry $\varphi: V \to V$, defined by
\begin{equation}
  \label{eq:isometry}
    \varphi(v_i) = v_i \qquad \varphi(u_i) = u_i - s_i -\half
    \sum_{j=1}^r \left<s_i,s_j\right> v_j \qquad  \varphi(x) = x +
    \sum_{i=1}^r \left<s_i,x\right> v_i,
\end{equation}
for $x\in W$ and for some $s_i \in W$.  (This is obtained by
extending the linear map $v_i \to v_i$ and $u_i \mapsto u_i - s_i$ to
an isometry of $V$.)  Under $\varphi$ the 3-brackets
\eqref{eq:3-brackets} take the following form
\begin{equation}
  \label{eq:3-brackets-phi}
  \begin{aligned}[m]
    [\varphi(u_i),\varphi(u_j),\varphi(u_k)] &= \varphi(K^\varphi_{ijk}) + \sum_{\ell=1}^r L^\varphi_{ijk\ell} v_\ell\\ 
    [\varphi(u_i),\varphi(u_j),\varphi(x)] &= \varphi(J^\varphi_{ij} x) - \sum_{k=1}^r \left<K^\varphi_{ijk},x\right> v_k\\
    [\varphi(u_i),\varphi(x),\varphi(y)] &= \varphi([x,y]^\varphi_i) - \sum_{j=1}^r \left<x,J^\varphi_{ij} y\right> v_j\\
    [\varphi(x),\varphi(y),\varphi(z)] &= \varphi([x,y,z]_W) - \sum_{i=1}^r \left<[x,y]^\varphi_i,z\right> v_i,
  \end{aligned}
\end{equation}
where
\begin{equation}
  \label{eq:isometry-changes}
  \begin{aligned}[m]
    [x,y]^\varphi_i &= [x,y]_i + [s_i,x,y]_W\\
    J^\varphi_{ij} x &= J_{ij} x + [s_i,x]_j - [s_j,x]_i + [s_i,s_j,x]_W\\
    K^\varphi_{ijk} & = K_{ijk} - J_{ij} s_k - J_{jk} s_i - J_{ki} s_j
    + [s_i,s_j]_k + [s_j,s_k]_i + [s_k,s_i]_j - [s_i,s_j,s_k]_W\\
    L^\varphi_{ijk\ell} &= L_{ijk\ell} + \left<K_{jk\ell},s_i\right> -
    \left<K_{k\ell i},s_j\right> + \left<K_{\ell ij},s_k\right> -
    \left<K_{ijk},s_\ell\right>\\
    &\quad - \left<s_i,J_{k\ell}s_j\right> -
    \left<s_k,J_{j\ell}s_i\right> - \left<s_j,J_{i\ell}s_k\right> +
    \left<s_\ell,J_{jk}s_i\right> + \left<s_\ell,J_{ki}s_j\right> +
    \left<s_\ell,J_{ij}s_k\right>\\
    &\quad + \left<[s_i,s_j]_\ell, s_k\right> - \left<[s_i,s_j]_k,
      s_\ell\right> - \left<[s_k,s_i]_j, s_\ell\right> -
    \left<[s_j,s_k]_i, s_\ell\right> +
    \left<[s_i,s_j,s_k]_W,s_\ell\right>.
  \end{aligned}
\end{equation}

\begin{lemma}
  There exists $s_i\in S$ such that the following conditions are met
  for all $x \in S$:
  \begin{equation}
    [x,-]^\varphi_i = 0 \qquad J^\varphi_{ij} x = 0 \qquad
    \left<K^\varphi_{ijk},x\right> = 0.
  \end{equation}
\end{lemma}

Assuming for a moment that this is the case, the only nonzero
3-brackets involving elements in $\varphi(S)$ are
\begin{equation}
  [\varphi(x),\varphi(y),\varphi(z)] = \varphi([x,y,z]_W),
\end{equation}
and this means that $\varphi(S)$ is a nondegenerate ideal of $V$,
whence $V = \varphi(S) \oplus \varphi(S)^\perp$.  But this violates
the indecomposability of $V$, unless $S=0$.

\begin{proof}[Proof of the lemma]
  To show the existence of the $s_i$, let us decompose $S = S_4^{(1)} \oplus \dots \oplus S_4^{(m)}$ into $m$ copies of the unique
  simple positive-definite 3-Lie algebra $S_4$.  As shown in \cite[§3.2]{Lor3Lie}, since $J_{ij}$ and $[x,-]_i$ define
  skewsymmetric derivations of $W$, they preserve the decomposition of $W$ into $S\oplus A$ and that of $S$ into its simple
  factors.  One consequence of this fact is that $J_{ij}x \in S$ for all $x \in S$ and $[x,y]_i \in S$ for all $x,y\in S$, and
  similarly if we substitute $S$ for any of its simple factors in the previous statement.  Notice in addition that putting $i=j$
  in equation \eqref{eq:C2-W}, $[-,-]_i$ obeys the Jacobi identity.  Hence on any one of the simple factors of $S$ --- let's call
  it generically $S_4$ --- the bracket $[-,-]_i$ defines the structure of a four-dimensional Lie algebra.  This Lie algebra is
  metric by equation \eqref{eq:metricity-i} and positive definite.  There are (up to isomorphism) precisely two four-dimensional
  positive-definite metric Lie algebras: the abelian Lie algebra and $\fso(3) \oplus \RR$.  In either case, as shown in
  \cite[§3.2]{Lor3Lie}, there exists a unique $s_i \in S_4$ such that $[s_i,x,y]_W = [x,y]_i$ for $x,y\in S_4$.  (In the former
  case, $s_i=0$.)  Since this is true for all simple factors, we conclude that there exists $s_i \in S$ such that $[s_i,x,y]_W =
  [x,y]_i$ for $x,y\in S$ and for all $i$.

  Now equation \eqref{eq:C2-W} says that for all $x,y,z\in S$,
  \begin{align*}
    [J_{ij}x,y,z]_W &= [[x,y]_i,z]_j + [[y,z]_j,x]_i + [[z,x]_i,y]_j\\
    &= [s_j,[s_i, x,y]_W,z]_W + [s_i,[s_j,y,z]_W,x]_W + [s_j,[s_i, z,x]_W,y]_W\\
    &= [[s_i,s_j,x]_W,y,x]_W, &&\text{using \eqref{eq:W-3la}}
  \end{align*}
  which implies that $J_{ij}x - [s_i,s_j,x]_W$ centralises $S$, and thus is in $A$.  However, for $x \in S$, both $J_{ij}x \in S$
  and $[s_i,s_j,x]_W \in S$, so that $J_{ij}x = [s_i,s_j,x]_W$.  Similarly, equation \eqref{eq:C1-W} says that for all $x,y \in S$,
  \begin{align*}
    [x,y,K_{ijk}]_W &= J_{ij} [x,y]_k - [J_{ij} x,y ]_k - [x, J_{ij} y]_k\\
    &= [s_i,s_j, [s_k,x,y]_W]_W - [s_k,[s_i,s_j,x]_W,y ]_W - [s_k, x, [s_i,s_j y]_W]_W\\
    &= [[s_i,s_j,s_k]_W,x,y]_W, && \text{using \eqref{eq:W-3la}}
  \end{align*}
  which implies that $K_{ijk}-[s_i,s_j,s_k]_W$ centralises $S$, whence $K_{ijk} -
  [s_i,s_j,s_k]_W  = K^A_{ijk} \in A$.  Finally, using the explicit formulae for
  $J^\varphi_{ij}$ and $K^\varphi_{ijk}$ in equation
  \eqref{eq:isometry-changes}, we see that for all all $x \in S$,
  \begin{align*}
    J^\varphi_{ij} x &= J_{ij} x + [s_i,x]_j - [s_j,x]_i + [s_i,s_j,x]_W\\
    &= [s_i,s_j,x]_W + [s_j,s_i,x]_W - [s_i,s_j,x]_W + [s_i,s_j,x]_W = 0\\
    \intertext{and}
    K^\varphi_{ijk} &= K_{ijk} - J_{ij} s_k - J_{jk} s_i - J_{ki} s_j +
    [s_i,s_j]_k + [s_j,s_k]_i + [s_k,s_i]_j - [s_i,s_j,s_k]_W\\
    &= K^A_{ijk} + [s_i,s_j,s_k]_W - [s_i,s_j,s_k]_W - [s_j,s_k,s_i]_W -
    [s_k,s_i,s_j]_W\\
    & \quad +  [s_k, s_i,s_j]_W + [s_i,s_j,s_k]_W +
    [s_j,s_k,s_i]_W - [s_i,s_j,s_k]_W= K^A_{ijk},
  \end{align*}
  whence $\left<K^\varphi_{ijk},x\right> = 0$ for all $x \in S$.
\end{proof}

We may summarise the above discussion as follows.

\begin{lemma}
  Let $V$ be a finite-dimensional indecomposable metric 3-Lie algebra
  of index $r>0$ with a maximally isotropic centre.  Then as a vector
  space
  \begin{equation}
    V = \bigoplus_{i=1}^r \left( \RR u_i \oplus \RR v_i\right) \oplus
    W,
  \end{equation}
  where $W$ is positive-definite, $u_i,v_i \perp W$,
  $\left<u_i,u_j\right> = 0$, $\left<v_i,v_j\right> = 0$ and
  $\left<u_i,v_j\right>= \delta_{ij}$.  The $v_i$ span the maximally
  isotropic centre.  The nonzero 3-brackets are given by
  \begin{equation}
    \label{eq:V3b}
    \begin{aligned}[m]
      [u_i,u_j,u_k] &= K_{ijk} + \sum_{\ell=1}^r L_{ijk\ell} v_\ell\\ 
      [u_i,u_j,x] &= J_{ij} x - \sum_{k=1}^r \left<K_{ijk},x\right> v_k\\
      [u_i,x,y] &= [x,y]_i - \sum_{j=1}^r \left<x,J_{ij} y\right> v_j\\
      [x,y,z] &= - \sum_{i=1}^r \left<[x,y]_i,z\right> v_i,
    \end{aligned}
  \end{equation}
  for all $x,y,z\in W$ and for some $L_{ijk\ell} \in \RR$, $K_{ijk}
  \in W$, $J_{ij} \in \fso(W)$, all of which are totally skewsymmetric
  in their indices, and bilinear alternating brackets $[-,-]_i : W
  \times W \to W$ satisfying equation \eqref{eq:metricity-i}.
  Furthermore, the fundamental identity of the 3-brackets
  \eqref{eq:V3b} is equivalent to the following conditions on
  $K_{ijk}$, $J_{ij}$ and $[-,-]_i$:
  \begin{subequations}\label{eq:FI-V3b}
    \begin{align}
      [x, [y,z]_i]_j &=  [[x,y]_j,z]_i + [y,[x,z]_j]_i \label{eq:C2}\\
      [[x,y]_i,z]_j &=  [[x,y]_j,z]_i  \label{eq:C3}\\
      J_{ij} [x,y]_k &= [J_{ij} x,y ]_k + [x, J_{ij} y ]_k \label{eq:C1}\\
      0 &= J_{j\ell} [x,y]_i + J_{\ell i}[x,y]_j + J_{ij}[x,y]_\ell \label{eq:C6}\\
      [K_{ijk},x]_\ell - [K_{ij\ell},x]_k &= \left(J_{ij}J_{k\ell} - J_{k\ell}J_{ij}\right) x  \label{eq:C4}\\
      [x,K_{jk\ell}]_i &= \left(J_{jk}J_{i\ell} + J_{k\ell}J_{ij} + J_{j\ell}J_{ki}\right) x  \label{eq:C5}\\
      J_{ij}K_{k\ell m} &= J_{\ell m}K_{ijk}  + J_{mk}K_{ij\ell} + J_{k\ell}K_{ijm}  \label{eq:C7}\\
      0 &= \left<K_{ijn},K_{k\ell m}\right> + \left<K_{ij\ell},K_{mnk}\right> - \left<K_{ijm},K_{nk\ell}\right> - \left<K_{ijk},K_{\ell mn}\right> . \label{eq:C8}
    \end{align}
  \end{subequations}
\end{lemma}

There are less equations in \eqref{eq:FI-V3b} than are obtained from \eqref{eq:FI-V-pre} by simply making $W$ abelian.  It is not
hard to show that the equations in \eqref{eq:FI-V3b} imply the rest.  The study of equations \eqref{eq:FI-V3b} will take us until
the end of this section.  The analysis of these conditions will break naturally into several steps.  In the first step we will
solve equations \eqref{eq:C2} and \eqref{eq:C3} for the $[-,-]_i$.  We will then solve equations \eqref{eq:C1} and \eqref{eq:C6},
which will turn allow us to solve equations \eqref{eq:C4} and \eqref{eq:C5} for the $J_{ij}$.  Finally we will solve equation
\eqref{eq:C7}.  We will not solve equation \eqref{eq:C8}.  In fact, this equation defines an algebraic variety (an intersection of
conics) which parametrises these 3-algebras.

\subsubsection{Solving for the $[-,-]_i$}
\label{sec:solving-brackets}

Condition \eqref{eq:C2} for $i=j$ says that $[-,-]_i$ defines a Lie
algebra structure on $W$, denoted $\fg_i$.  By equation
\eqref{eq:metricity-i}, $\fg_i$ is a metric Lie algebra.  Since the
inner product on $W$ is positive-definite, $\fg_i$ is reductive,
whence $\fg_i = [\fg_i,\fg_i] \oplus \fz_i$, where $\fs_i :=
[\fg_i,\fg_i]$ is the semisimple derived ideal of $\fg_i$ and $\fz_i$
is the centre of $\fg_i$.  The following lemma will prove useful.

\begin{lemma}
  \label{lem:g1-simple}
  Let $\fg_i$, $i=1,\dots,r$, be a family of reductive Lie algebras
  sharing the same underlying vector space $W$ and let $[-,-]_i$
  denote the Lie bracket of $\fg_i$.  Suppose that they satisfy
  equations \eqref{eq:C2} and \eqref{eq:C3} and in addition that one
  of these Lie algebras, $\fg_1$ say, is simple.  Then for all $x,y\in
  W$,
  \begin{equation}
    [x,y]_i = \kappa_i [x,y]_1,
  \end{equation}
  where $\kappa_i \in \RR$.
\end{lemma}

\begin{proof}
  Equation \eqref{eq:C2} says that for all $x \in W$, $\ad_i x
  :=[x,-]_i$ is a derivation of $\fg_j$, for all $i,j$.  In
  particular, $\ad_1 x$ is a derivation of $\fg_i$.  Since derivations
  preserve the centre, $\ad_1 x : \fz_i \to \fz_i$, whence the
  subspace $\fz_i$ is an ideal of $\fg_1$.  Since by hypothesis,
  $\fg_1$ is simple, we must have that either $\fz_i = W$, in which
  case $\fg_i$ is abelian and the lemma holds with $\kappa_i=0$, or
  else $\fz_i=0$, in which case $\fg_i$ is semisimple.  It remains
  therefore to study this case.

  Equation \eqref{eq:C2} again says that $\ad_i x$ is a derivation of
  $\fg_1$.  Since all derivations of $\fg_1$ are inner, this means
  that there is some element $y$ such that $\ad_i x = \ad_1 y$.  This
  element is moreover unique because $\ad_1$ has trivial kernel.  In
  other words, this defines a linear map
  \begin{equation}
    \label{eq:psi-i}
    \psi_i: \fg_i \to \fg_1 \qquad\text{by}\qquad \ad_i x = \ad_1
    \psi_i x \qquad \forall x \in W.
  \end{equation}
  This linear map is a vector space isomorphism since $\ker \psi_i
  \subset \ker \ad_i = 0$, for $\fg_i$ semisimple.  Now suppose that
  $I \lhd \fg_i$ is an ideal, whence $\ad_i(x) I \subset I$ for all $x
  \in \fg_i$.  This means that $\ad_1(y) I \subset I$ for all $y \in
  \fg_1$, whence $I$ is also an ideal of $\fg_1$.  Since $\fg_1$ is
  simple, this means that $I=0$ or else $I=W$; in other words, $\fg_i$
  is simple.

  Now for all $x,y,z \in W$, we have
  \begin{align*}
    [\psi_i[x,y]_i,z]_1 &= [[x,y]_i,z]_i && \text{by equation \eqref{eq:psi-i}}\\
    &= [x,[y,z]_i]_i - [y,[x,z]_i]_i && \text{by the Jacobi identity of $\fg_i$}\\
    &= [\psi_ix,[\psi_iy,z]_1]_1 - [\psi_iy,[\psi_ix,z]_1]_1 && \text{by  equation \eqref{eq:psi-i}}\\
    &= [[\psi_i x, \psi_i y]_1, z]_1 && \text{by the Jacobi identity of $\fg_1$}
  \end{align*}
  and since $\fg_1$ has trivial centre, we conclude that
  \begin{equation*}
    \psi_i[x,y]_i = [\psi_i x,\psi_i y]_1,
  \end{equation*}
  whence $\psi_i: \fg_i \to \fg_1$ is a Lie algebra isomorphism.

  Next, condition \eqref{eq:C3} says that $\ad_1 [x,y]_i = \ad_i
  [x,y]_1$, whence using equation \eqref{eq:psi-i}, we find that
  $\ad_1 [x,y]_i = \ad_1 \psi_i [x,y]_1$, and since $\ad_1$ has
  trivial kernel,  $[x,y]_i = \psi_i [x,y]_1$.  We may rewrite this
  equation as $\ad_i x = \psi_i \ad_1 x$ for all $x$, which again by
  virtue of \eqref{eq:psi-i}, becomes $\ad_1 \psi_i x = \psi_i \ad_1
  x$, whence $\psi_i$ commutes with the adjoint representation of
  $\fg_1$.  Since $\fg_1$ is simple, Schur's Lemma says that
  $\psi_i$ must be a multiple, $\kappa_i$ say, of the identity.  In
  other words,  $\ad_i x = \kappa_i \ad_1 x$, which proves the lemma.
\end{proof}

Let us now consider the general case when none of the $\fg_i$ are
simple.  Let us focus on two reductive Lie algebras, $\fg_i = \fz_i
\oplus \fs_i$, for $i=1,2$ say, sharing the same underlying vector
space $W$.  We will further decompose $\fs_i$ into its simple ideals
\begin{equation}
  \fs_i = \bigoplus_{\alpha =1}^{N_i} \fs_i^\alpha.
\end{equation}
For every $x \in W$, $\ad_1 x$ is a derivation of $\fg_2$, whence it
preserves the centre $\fz_2$ and each simple ideal $\fs_2^\beta$.
This means that $\fz_2$ and $\fs_2^\beta$ are themselves ideals of
$\fg_1$, whence
\begin{equation}
  \fz_2 = E_0 \oplus \bigoplus_{\alpha \in I_0} \fs_1^\alpha
  \qquad\text{and}\qquad
  \fs_2^\beta = E_\beta \oplus \bigoplus_{\alpha \in I_\beta}
  \fs_1^\alpha \qquad \forall \beta\in\left\{1,2,\dots,N_2\right\},
\end{equation}
and where the index sets $I_0,I_1,\dots,I_{N_2}$ define a partition of
$\left\{1,\dots,N_1\right\}$, and
\begin{equation}
  \fz_1 = E_0 \oplus E_1 \oplus \cdots \oplus E_{N_2}
\end{equation}
is an orthogonal decomposition of $\fz_1$.  But now notice that the
restriction of $\fg_1$ to $E_\beta \oplus \bigoplus_{\alpha \in
  I_\beta} \fs_1^\alpha$ is reductive, whence we may apply
Lemma~\ref{lem:g1-simple} to each simple $\fs_2^\beta$ in turn.  This
allows us to conclude that for each $\beta$, either $\fs_2^\beta =
E_\beta$ or else $\fs_2^\beta = \fs_1^\alpha$, for some
$\alpha \in \left\{1,2,\dots, N_1\right\}$ which depends on $\beta$,
and in this latter case, $[x,y]_{\fs_2^\beta} = \kappa
[x,y]_{\fs_1^{\alpha}}$, for some nonzero constant $\kappa$.

This means that, given any one Lie algebra $\fg_i$, any other Lie
algebra $\fg_j$ in the same family is obtained by multiplying its
simple factors by some constants (which may be different in each
factor and may also be zero) and maybe promoting part of its centre to
be semisimple.

The metric Lie algebras $\fg_i$ induce the following orthogonal
decomposition of the underlying vector space $W$.  We let $W_0 =
\bigcap_{i=1}^r \fz_i$ be the intersection of all the centres of the
reductive Lie algebras $\fg_i$.  Then we have the following orthogonal
direct sum $W = W_0 \oplus \bigoplus_{\alpha =1}^N W_\alpha$, where
restricted to each $W_{\alpha>0}$ at least one of the Lie algebras,
$\fg_i$ say, is simple and hence all other Lie algebras $\fg_{j\neq
  i}$ are such that for all $x,y\in W_\alpha$,
\begin{equation}
  [x,y]_j = \kappa_{ij}^\alpha [x,y]_i \qquad \exists
  \kappa_{ij}^\alpha \in \RR.
\end{equation}

To simplify the notation, we define a semisimple Lie algebra structure
$\fg$ on the perpendicular complement of $W_0$, whose Lie bracket
$[-,-]$ is defined in such a way that for all $x,y\in
W_\alpha$, $[x,y] := [x,y]_i$, where $i\in\{1,2,\dots,r\}$ is the
smallest such integer for which the restriction of $\fg_i$ to
$W_\alpha$ is simple. (That such an integer $i$ exists follows from
the definition of $W_0$ and of the $W_\alpha$.)   It then follows that
the restriction to $W_\alpha$ of every other $\fg_{j\neq i}$ is a
(possibly zero) multiple of $\fg$.

We summarise this discussion in the following lemma, which summarises
the solution of equations \eqref{eq:C2} and \eqref{eq:C3}.

\begin{lemma}
  \label{lem:C2C3}
  Let $\fg_i$, $i=1,\dots,r$, be a family of metric Lie algebras
  sharing the same underlying euclidean vector space $W$ and let
  $[-,-]_i$ denote the Lie bracket of $\fg_i$.  Suppose that they
  satisfy equations \eqref{eq:C2} and \eqref{eq:C3}.  Then there is an
  orthogonal decomposition
  \begin{equation}
    \label{eq:W-decomp}
    W = W_0 \oplus \bigoplus_{\alpha =1}^N W_\alpha,
  \end{equation}
  where
  \begin{equation}
    \label{eq:g-brackets}
    [x,y]_i =
    \begin{cases}
      0 & \text{if $x,y\in W_0$;}\\
      \kappa_i^\alpha [x,y] & \text{if $x,y \in W_\alpha$,}
    \end{cases}
  \end{equation}
  for some $\kappa_i^\alpha \in \RR$ and where $[-,-]$ are the Lie
  brackets of a semisimple Lie algebra $\fg$ with underlying vector
  space $\bigoplus_{\alpha =1}^N W_\alpha$.
\end{lemma}

\subsubsection{Solving for the $J_{ij}$}
\label{sec:solving-Js}

Next we study the equations \eqref{eq:C1} and \eqref{eq:C6}, which
involve only $J_{ij}$.  Equation \eqref{eq:C1} says that each $J_{ij}$
is a derivation over the $\fg_k$ for all $i,j,k$.  Since derivations
preserve the centre, every $J_{ij}$ preserves the centre of every
$\fg_k$ and hence it preserves their intersection $W_0$.  Since
$J_{ij}$ preserves the inner product, it also preserves the
perpendicular complement of $W_0$ in $W$, which is the underlying
vector space of the semisimple Lie algebra $\fg$ of the previous
lemma.  Equation \eqref{eq:C1} does not constrain the component of
$J_{ij}$ acting on $W_0$ since all the $[-,-]_k$ vanish there, but it
does constrain the components of $J_{ij}$ acting on $\bigoplus_{\alpha
  =1}^N W_\alpha$.  Fix some $\alpha$ and let $x,y\in W_\alpha$.  Then
by virtue of equation \eqref{eq:g-brackets}, equation \eqref{eq:C1}
says that
\begin{equation}
  \kappa_k^\alpha \left( J_{ij} [x,y] - [J_{ij}x,y] - [x,J_{ij}y]
  \right) = 0.
\end{equation}
Since, given any $\alpha$ there will be at least some $k$ for which
$\kappa_k^\alpha \neq 0$, we see that $J_{ij}$ is a derivation of
$\fg$.  Since $\fg$ is semisimple, this derivation is inner, where
there exists a unique $z_{ij} \in \fg$, such that $J_{ij} y =
[z_{ij},y]$ for all $y \in \fg$.  Since the simple ideals of $\fg$ are
submodules under the adjoint representation, $J_{ij}$ preserves each
of the simple ideals and hence it preserves the decomposition
\eqref{eq:W-decomp}.  Let $z_{ij}^\alpha$ denote the component of
$z_{ij}$ along $W_\alpha$.  Equation \eqref{eq:C6} can now be
rewritten for $x,y\in W_\alpha$ as
\begin{equation}
  \kappa_i^\alpha [z_{j\ell}^\alpha, [x,y]] + \kappa_j^\alpha
  [z^\alpha_{\ell i},[x,y]] + \kappa_\ell^\alpha [z^\alpha_{ij},[x,y]]
  = 0.
\end{equation}
Since $\fg$ has trivial centre, this is equivalent to
\begin{equation}
  \label{eq:kwedgez=0}
  \kappa_i^\alpha z_{j\ell}^\alpha + \kappa_j^\alpha z^\alpha_{\ell i}
  + \kappa_\ell^\alpha z^\alpha_{ij} = 0,
\end{equation}
which can be written more suggestively as $\kappa^\alpha \wedge
z^\alpha = 0$, where $\kappa^\alpha \in \RR^r$ and $z^\alpha \in
\Lambda^2\RR^r \otimes W_\alpha$.  This equation has as unique
solution $z^\alpha = \kappa^\alpha \wedge s^\alpha$, for some $s^\alpha
\in \RR^r\otimes W_\alpha$, or in indices
\begin{equation}
  \label{eq:zijalpha}
  z^\alpha_{ij} = \kappa_i^\alpha s_j^\alpha - \kappa_j^\alpha
  s_i^\alpha \qquad \exists s_i^\alpha \in W_\alpha.
\end{equation}
Let $s_i = \sum_\alpha s_i^\alpha \in \fg$ and consider now the
isometry $\varphi: V \to V$ defined by
\begin{equation}
  \begin{aligned}[m]
    \varphi(v_i) &= v_i\\
    \varphi(z) &= z\\
    \varphi(u_i) &= u_i - s_i - \half \sum_j \left<s_i,s_j\right> v_j\\
    \varphi(x) &= x + \sum_i \left<s_i,x\right> v_i,
  \end{aligned}
\end{equation}
for all $z\in W_0$ and all $x\in \bigoplus_{\alpha =1}^N W_\alpha$.
The effect of such a transformation on the 3-brackets \eqref{eq:V3b}
is an uninteresting modification of $K_{ijk}$ and $L_{ijk\ell}$
and the more interesting disappearance of $J_{ij}$ from the 3-brackets
involving elements in $W_\alpha$.  Indeed, for all $x \in W_\alpha$,
we have
\begin{align*}
  [\varphi(u_i),\varphi(u_j),\varphi(x)] &= [u_i-s_i,u_j-s_j,x]\\
  &= [u_i,u_j,x] + [u_j,s_i,x] - [u_i,s_j,x] + [s_i,s_j,x]\\
  &= J_{ij} x + [s_i,x]_j - [s_j,x]_i + \text{central terms}\\
  &= [z_{ij}^\alpha, x] + \kappa_j^\alpha [s_i^\alpha,x] -
  \kappa_i^\alpha [s_j^\alpha,x] + \text{central terms}\\
  &= [z_{ij}^\alpha + \kappa_j^\alpha s_i^\alpha - \kappa_i^\alpha
  s_j^\alpha , x] + \text{central terms}\\
  &= 0 + \text{central terms},
\end{align*}
where we have used equation \eqref{eq:zijalpha}.

This means that without loss of generality we may assume that $J_{ij}
x =0$ for all $x \in W_\alpha$ for any $\alpha$.  Now consider
equation \eqref{eq:C5} for $x \in \bigoplus_{\alpha =1}^N W_\alpha$.
The right-hand side vanishes, whence $[K_{ijk},x]_\ell=0$.  Also if $x
\in W_0$, then $[K_{ijk},x]_\ell=0$ because $x$ is central with
respect to all $\fg_\ell$.  Therefore we see that $K_{ijk}$ is central
with respect to all $\fg_\ell$, and hence $K_{ijk} \in W_0$.

In other words, we have proved the following

\begin{lemma}
  \label{lem:C1C6C5}
  In the notation of Lemma~\ref{lem:C2C3}, the nonzero 3-brackets for
  $V$ may be brought to the form
  \begin{equation}
    \label{eq:V3b2}
    \begin{aligned}[m]
      [u_i,u_j,u_k] &= K_{ijk} + \sum_{\ell=1}^r L_{ijk\ell} v_\ell\\
      [u_i,u_j,x_0] &= J_{ij} x_0 - \sum_{k=1}^r \left<K_{ijk},x_0\right> v_k\\
      [u_i,x_0,y_0] &= - \sum_{j=1}^r \left<x_0,J_{ij} y_0\right> v_j\\
      [u_i,x_\alpha,y_\alpha] &= \kappa_i^\alpha [x,y]\\
      [x_\alpha,y_\alpha,z_\alpha] &= -
      \left<[x_\alpha,y_\alpha],z_\alpha\right> \sum_{i=1}^r
      \kappa_i^\alpha v_i,
    \end{aligned}
  \end{equation}
  for all $x_\alpha,y_\alpha,z_\alpha\in W_\alpha$, $x_0,y_0 \in W_0$
  and for some $L_{ijk\ell} \in \RR$, $K_{ijk} \in W_0$ and
  $J_{ij} \in \fso(W_0)$, all of which are totally skewsymmetric in
  their indices.
\end{lemma}

Since their left-hand sides vanish, equations \eqref{eq:C4} and
\eqref{eq:C5} become conditions on $J_{ij} \in \fso(W_0)$:
\begin{align}
  J_{ij}J_{k\ell} - J_{k\ell}J_{ij} &= 0, \label{eq:JJ1}\\
  J_{jk}J_{i\ell} + J_{k\ell}J_{ij} + J_{j\ell}J_{ki} &= 0. \label{eq:JJ2}
\end{align}
The first condition says that the $J_{ij}$ commute, whence since the inner product on $W_0$ is positive-definite, they must belong
to the same Cartan subalgebra $\fh \subset \fso(W_0)$.  Let $H_\pi$, for $\pi=1,\dots,\lfloor\frac{\dim W_0}{2}\rfloor$, denote a
basis for $\fh$, with each $H_\pi$ corresponding to the generator of infinitesimal rotations in mutually orthogonal 2-planes in
$W_0$.  In particular, this means that $H_\pi H_\varrho = 0$ for $\pi \neq \varrho$ and that $H_\pi^2 = - \Pi_\pi$, with $\Pi_\pi$
the orthogonal projector onto the 2-plane labelled by $\pi$.  We write $J_{ij}^\pi \in \RR$ for the component of $J_{ij}$ along
$H_\pi$.  Fixing $\pi$ we may think of $J_{ij}^\pi$ as the components of $J^\pi \in \Lambda^2 \RR^r$.  Using the relations obeyed
by the $H_\pi$, equation \eqref{eq:JJ2} separates into $\lfloor\frac{\dim W_0}{2}\rfloor$ equations, one for each value of $\pi$,
which in terms of $J^\pi$ can be written simply as $J^\pi \wedge J^\pi = 0$.  This is a special case of a Plücker relation and
says that $J^\pi$ is decomposable; that is, $J^\pi = \eta^\pi \wedge \zeta^\pi$ for some $\eta^\pi, \zeta^\pi \in \RR^r$.  In
other words, the solution of equations \eqref{eq:JJ1} and \eqref{eq:JJ2} is
\begin{equation}
  J_{ij} = \sum_\pi \left(\eta_i^\pi \zeta^\pi_j - \eta_j^\pi \zeta^\pi_i\right)
  H_\pi
\end{equation}
living in a Cartan subalgebra $\fh \subset \fso(W_0)$.

\subsubsection{Solving for the $K_{ijk}$}
\label{sec:solving-K}

It remains to solve equations \eqref{eq:C7} and \eqref{eq:C8} for
$K_{ijk}$.  We shall concentrate on the linear equation
\eqref{eq:C7}.  This is a linear equation on $K \in \Lambda^3\RR^r
\otimes W_0$ and says that it is in the kernel of a linear map
\begin{equation}
  \label{eq:cocycle}
  \begin{CD}
    \Lambda^3\RR^r \otimes W_0 @>>> \Lambda^2\RR^r \otimes
    \Lambda^3\RR^r \otimes W_0
  \end{CD}
\end{equation}
defined by
\begin{equation}
\label{eq:cocycle-eqn}
  K_{ijk} \mapsto J_{ij} K_{k\ell m} - J_{\ell m}K_{ijk}  - J_{mk}K_{ij\ell} - J_{k\ell}K_{ijm}.
\end{equation}
The expression in the right-hand side is manifestly skewsymmetric in
$ij$ and $k\ell m$ separately, whence it belongs to $\Lambda^2\RR^r
\otimes \Lambda^3\RR^r \otimes W_0$ as stated above.  For generic $r$
(here $r\geq 5$) we may decompose
\begin{equation}
  \Lambda^2\RR^r \otimes \Lambda^3\RR^r = Y^{\yng(2,2,1)}\RR^r \oplus
  Y^{\yng(2,1,1,1)}\RR^r \oplus \Lambda^5 \RR^r,
\end{equation}
where $Y^{\text{Young tableau}}$ denotes the corresponding Young
symmetriser representation.  Then one can see that the right-hand side
of \eqref{eq:cocycle-eqn} has no component in the first of the above
summands and hence lives in the remaining two summands, which are
isomorphic to $\RR^r \otimes \Lambda^4 \RR^r$.

We now observe that via an isometry of $V$ of the form
\begin{equation}
  \begin{aligned}[m]
    \varphi(v_i) &= v_i\\
    \varphi(x_\alpha) &= x_\alpha\\
    \varphi(u_i) &= u_i + t_i - \half \sum_j \left<t_i,t_j\right> v_j\\
    \varphi(x_0) &= x_0 - \sum_i \left<x_0, t_i\right> v_i,
  \end{aligned}
\end{equation}
for $t_i \in W_0$, the form of the 3-brackets \eqref{eq:V3b2} remains
invariant, but with $K_{ijk}$ and $L_{ijk\ell}$ transforming by
\begin{align}
  K_{ijk} &\mapsto K_{ijk} + J_{ij} t_k + J_{jk} t_i + J_{ki} t_j,\\
\intertext{and}
  \begin{split}
    L_{ijk\ell} &\mapsto L_{ijk\ell} +
    \left<K_{ijk},t_\ell\right> - \left<K_{\ell ij},t_k\right> +
    \left<K_{k\ell i},t_j\right> - \left<K_{jk\ell},t_i\right>\\
    &\qquad + \left<J_{ij} t_k,t_\ell\right> + \left<J_{ki}
      t_j,t_\ell\right> + \left<J_{jk} t_i,t_\ell\right> +
    \left<J_{i\ell} t_j,t_k\right> + \left<J_{j\ell} t_k,t_i\right> +
    \left<J_{k\ell} t_i,t_j\right>,
  \end{split}
\end{align}
respectively.  In particular, this means that there is an ambiguity in
$K_{ijk}$, which can be thought of as shifting it by the image of the
linear map
\begin{equation}
  \label{eq:coboundary}
  \begin{CD}
    \RR^r \otimes W_0 @>>> \Lambda^3\RR^r \otimes W_0
  \end{CD}
\end{equation}
defined by
\begin{equation}
  t_i \mapsto J_{ij} t_k + J_{jk} t_i + J_{ki} t_j.
\end{equation}
The two maps \eqref{eq:cocycle} and \eqref{eq:coboundary} fit together
in a complex
\begin{equation}
  \label{eq:complex}
  \begin{CD}
    \RR^r \otimes W_0 @>>> \Lambda^3\RR^r \otimes W_0 @>>> \RR^r \otimes
    \Lambda^4\RR^r \otimes W_0,
  \end{CD}
\end{equation}
where the composition vanishes \emph{precisely} by virtue of equations
\eqref{eq:JJ1} and \eqref{eq:JJ2}.  We will show that this complex is
acyclic away from the kernel of $J$, which will mean that without loss
of generality we can take $K_{ijk}$ in the kernel of $J$ subject to
the final quadratic equation \eqref{eq:C8}.

Let us decompose $W_0$ into an orthogonal direct sum
\begin{equation}
  W_0 =
  \begin{cases}
     \bigoplus\limits_{\pi=1}^{(\dim W_0)/2} E_\pi, & \text{if $\dim W_0$ is even, and}\\[18pt]
     \RR w \oplus \bigoplus\limits_{\pi=1}^{(\dim W_0-1)/2} E_\pi, & \text{if $\dim W_0$ is odd,}
  \end{cases}
\end{equation}
where $E_\pi$ are mutually orthogonal 2-planes and, in the second case,
$w$ is a vector perpendicular to all of them.  On $E_\pi$ the Cartan
generator $H_\pi$ acts as a complex structure, and hence we may identify
each $E_\pi$ with a complex one-dimensional vector space and $H_\pi$ with
multiplication by $i$.  This decomposition of $W_\pi$ allows us to
decompose $K_{ijk} = K^w_{ijk} + \sum_\pi K^\pi_{ijk}$, where the first
term is there only in the odd-dimensional situation and the
$K^\pi_{ijk}$ are complex numbers.  The complex \eqref{eq:complex}
breaks up into $\lfloor\frac{\dim W_0}{2}\rfloor$ complexes, one for
each value of $\pi$.  If $J^\pi = 0$ then $K_{ijk}^\pi$ is not constrained
there, but if $J^\pi = \eta^\pi \wedge \zeta^\pi \neq 0$ the complex turns
out to have no homology, as we now show.

Without loss of generality we may choose the vectors $\eta^\pi$ and
$\zeta^\pi$ to be the elementary vectors $e_1$ and $e_2$ in $\RR^r$, so
that $J^\pi$ has a $J^\pi_{12}=1$ and all other $J^\pi_{ij}=0$.  Take $i=1$
and $j=2$ in the cocycle condition \eqref{eq:cocycle}, to obtain
\begin{equation}
  K^\pi_{k\ell m} = J^\pi_{\ell m} K^\pi_{12k} + J^\pi_{mk} K^\pi_{12\ell} +
  J^\pi_{k\ell} K^\pi_{12m}.
\end{equation}
It follows that if any two of $k,\ell,m >2$, then $K^\pi_{k\ell m} = 0$.
In particular $K^\pi_{1ij} = K^\pi_{2ij} = 0$ for all $i,j>2$, whence only
$K^\pi_{12k}$ for $k>2$ can be nonzero.  However for $k>2$, $K^\pi_{12k} =
J^\pi_{12} e_k$, with $e_k$ the $k$th elementary vector in $\RR^r$, and
hence $K^\pi_{12k}$ is in the image of the map \eqref{eq:coboundary};
that is, a coboundary.  This shows that we may assume without loss of
generality that $K^\pi_{ijk} = 0$.  In summary, the only components of
$K_{ijk}$ which survive are those in the kernel of all the $J_{ij}$.
It is therefore convenient to split $W_0$ into an orthogonal direct
sum
\begin{equation}
  W_0 = E_0 \oplus \bigoplus_\pi E_\pi,
\end{equation}
where on each 2-plane $E_\pi$, $J^\pi = \eta^\pi \wedge \zeta^\pi \neq 0$,
whereas $J_{ij} x = 0$ for all $x \in E_0$.  Then we can take $K_{ijk}
\in E_0$.

Finally it remains to study the quadratic equation \eqref{eq:C8}.
First of all we mention that this equation is automatically satisfied
for $r\leq 4$.  To see this notice that the equation is skewsymmetric
in $k,\ell, m,n$, whence if $r<4$ it is automatically zero.  When $r=4$,
we have to take $k,\ell,m,n$ all different and hence the equation
becomes
\begin{equation*}
  \left<K_{ij1},K_{234}\right> - \left<K_{ij2},K_{341}\right> +
  \left<K_{ij3},K_{412}\right> - \left<K_{ij4},K_{123}\right> = 0,
\end{equation*}
which is skewsymmetric in $i,j$.  There are six possible choices for
$i,j$ but by symmetry any choice is equal to any other up to
relabeling, so without loss of generality let us take $i=1$ and $j=2$,
whence the first two terms are identically zero and the two remaining
terms satisfy
\begin{equation*}
  \left<K_{123},K_{412}\right> - \left<K_{124},K_{123}\right> = 0,
\end{equation*}
which is identically true.  This means that the cases of index $3$
and $4$ are classifiable using our results.  By contrast, the case of
index $5$ and above seems not to be tame.  An example should
suffice.  So let us take the case of $r=5$ and $\dim E_0 = 1$, so that
the $K_{ijk}$ can be taken to be real numbers.  The solutions to
\eqref{eq:C8} now describe the intersection of five quadrics in
$\RR^{10}$:
\begin{gather*}
K_{125} K_{134} - K_{124} K_{135} + K_{123} K_{145} = 0\\
K_{125} K_{234} - K_{124} K_{235} + K_{123} K_{245} = 0\\
K_{135} K_{234} - K_{134} K_{235} + K_{123} K_{345} = 0\\
K_{145} K_{234} - K_{134} K_{245} + K_{124} K_{345} = 0\\
K_{145} K_{235} - K_{135} K_{245} + K_{125} K_{345} = 0,
\end{gather*}
whence the solutions define an algebraic variety.  One possible branch
is given by setting $K_{1ij}=0$ for all $i,j$, which leaves
undetermined $K_{234}$, $K_{235}$, $K_{245}$ and $K_{345}$.  There are
other branches which are linearly related to this one: for
instance, setting $K_{2ij}=0$, et cetera, but there are also other
branches which are not linearly related to it.

\subsubsection{Summary and conclusions}
\label{sec:summary-conclusion}

Let us summarise the above results in terms of the following structure
theorem.

\begin{theorem}
  \label{thm:main}
  Let $V$ be a finite-dimensional indecomposable metric 3-Lie algebra
  of index $r>0$ with a maximally isotropic centre.  Then $V$ admits a
  vector space decomposition into $r+M+N+1$ orthogonal subspaces
  \begin{equation}
    \label{eq:decomp-V}
    V = \bigoplus_{i=1}^r \left(\RR u_i \oplus \RR v_i\right) \oplus
    \bigoplus_{\alpha =1}^N W_\alpha \oplus \bigoplus_{\pi=1}^M E_\pi
    \oplus E_0,
  \end{equation}
  where $W_\alpha$, $E_\pi$ and $E_0$ are positive-definite subspaces
  with the $E_\pi$ being two-dimensional, and where
  $\left<u_i,u_j\right> = \left<v_i,v_j\right> = 0$ and
  $\left<u_i,v_j\right> = \delta_{ij}$.  The 3-Lie algebra is defined
  in terms of the following data:
  \begin{itemize}
  \item $0 \neq \eta^\pi \wedge \zeta^\pi \in \Lambda^2 \RR^r$ for each $\pi=1,\dots,M$,
  \item $0 \neq \kappa^\alpha \in \RR^r$ for each $\alpha =1, \dots, N$,
  \item a metric simple Lie algebra structure $\fg_\alpha$ on each $W_\alpha$,
  \item $L \in \Lambda^4\RR^r$, and
  \item $K \in \Lambda^3\RR^r \otimes E_0$ subject to the equation
    \begin{equation*}
     \left<K_{ijn},K_{k\ell m}\right> + \left<K_{ij\ell},K_{mnk}\right> - \left<K_{ijm},K_{nk\ell}\right> - \left<K_{ijk},K_{\ell mn}\right> = 0,
    \end{equation*}
  \end{itemize}
  by the following 3-brackets, \footnote{We understand tacitly that if a 3-bracket is not listed here it vanishes.  Also every
    summation is written explicitly, so the summation convention is not in force.  In particular, there is no sum over $\pi$ in
    the third and fourth brackets.}
  \begin{equation}
    \label{eq:V3b-main}
    \begin{aligned}[m]
      [u_i,u_j,u_k] &= K_{ijk} + \sum_{\ell=1}^r L_{ijk\ell} v_\ell\\
      [u_i,u_j,x_0] &= - \sum_{k=1}^r \left<K_{ijk},x_0\right> v_k\\
      [u_i,u_j,x_\pi] &= J^\pi_{ij} H_\pi x_\pi\\
      [u_i,x_\pi,y_\pi] &= - \sum_{j=1}^r \left<x_\pi,J^\pi_{ij} H_\pi y_\pi\right> v_j\\
      [u_i,x_\alpha,y_\alpha] &= \kappa_i^\alpha [x_\alpha,y_\alpha]\\
      [x_\alpha,y_\alpha,z_\alpha] &= - \left<[x_\alpha,y_\alpha],z_\alpha\right> \sum_{i=1}^r \kappa_i^\alpha v_i,
    \end{aligned}
  \end{equation}
  for all $x_0 \in E_0$, $x_\pi,y_\pi \in E_\pi$ and $x_\alpha, y_\alpha, z_\alpha \in W_\alpha$, and where $J^\pi_{ij} =
  \eta^\pi_i \zeta^\pi_j - \eta^\pi_j \zeta^\pi_i$ and $H_\pi$ a complex structure on each 2-plane $E_\pi$.  The resulting 3-Lie
  algebra is indecomposable provided that there is no $x_0 \in E_0$ which is perpendicular to all the $K_{ijk}$, whence in
  particular $\dim E_0 \leq \binom{r}{3}$.
\end{theorem}

\subsection{Examples for low index}
\label{sec:examples}

Let us now show how to recover the known classifications in index
$\leq 2$ from Theorem~\ref{thm:main}.

Let us consider the case of minimal positive index $r=1$.  In that
case, the indices $i,j,k,l$ in Theorem~\ref{thm:main} can only take
the value $1$ and therefore $J_{ij}$, $K_{ijk}$ and $L_{ijkl}$ are not
present.  Indecomposability of $V$ forces $E_0=0$ and $E_\pi=0$,
whence letting $u=u_1$ and $v=v_1$, we have $V = \RR u \oplus \RR v
\oplus \bigoplus_{\alpha =1}^N W_\alpha$ as a vector space, with
$\left<u,u\right> = \left<v,v\right> = 0$, $\left<u,v\right> = 1$ and
$\bigoplus_{\alpha =1}^N W_\alpha$ euclidean.  The 3-brackets are:
\begin{equation}
  \begin{aligned}[m]
    [u,x_\alpha,y_\alpha] &= [x_\alpha,y_\alpha]\\
    [x_\alpha,y_\alpha,z_\alpha] &= - \left<[x_\alpha,y_\alpha],z_\alpha\right>  v,
  \end{aligned}
\end{equation}
for all $x_\alpha, y_\alpha, z_\alpha \in W_\alpha$ and where we have
redefined $ \kappa^\alpha [x_\alpha,y_\alpha] \to
[x_\alpha,y_\alpha]$, which is a simple Lie algebra on each
$W_\alpha$.  This agrees with the classification of lorentzian $3$-Lie
algebras in \cite{Lor3Lie} which was reviewed in the introduction.

Let us now consider $r=2$.  According to Theorem~\ref{thm:main}, those
with a maximally isotropic centre may now have a nonvanishing
$J_{12}$ while $K_{ijk}$ and $L_{ijkl}$ are still absent.
Indecomposability of $V$ forces $E_0=0$.  Therefore $W_0 = \bigoplus_{\pi=1}^M E_\pi$
and, as a vector space, $V = \RR u_1 \oplus \RR v_1 \oplus \RR u_2
\oplus \RR v_2 \oplus W_0 \oplus \bigoplus_{\alpha =1}^N W_\alpha$ with
$\left<u_i,u_j\right> = \left<v_i,v_j\right> = 0$,
$\left<u_i,v_j\right> = \delta_{ij}$, $\forall i,j = 1,2$ and $W_0
\oplus \bigoplus_{\alpha =1}^N W_\alpha$ is euclidean.  The $3$-brackets are
now:
\begin{equation} \label{eq:2pIndec}
  \begin{aligned}[m]
    [u_1,u_2,x_\pi] &= J x_\pi\\
    [u_1,x_\pi,y_\pi] &= - \left<x_\pi,J y_\pi\right> v_2\\
    [u_2,x_\pi,y_\pi] &=  \left<x_\pi,J y_\pi\right> v_1\\
    [u_1,x_\alpha,y_\alpha] &= \kappa_1^\alpha [x_\alpha,y_\alpha]\\
    [u_2,x_\alpha,y_\alpha] &= \kappa_2^\alpha [x_\alpha,y_\alpha]\\
    [x_\alpha,y_\alpha,z_\alpha] &= - \left<[x_\alpha,y_\alpha],z_\alpha\right>  \kappa_1^\alpha v_1 %
    - \left<[x_\alpha,y_\alpha],z_\alpha\right> \kappa_2^\alpha v_2,
  \end{aligned}
\end{equation}
for all $x_\pi,y_\pi \in E_\pi$ and $x_\alpha, y_\alpha, z_\alpha \in
W_\alpha$.  This agrees with the classification in \cite{2p3Lie} of
finite-dimensional indecomposable $3$-Lie algebras of index 2 whose
centre contains a maximally isotropic plane.  In that paper such
algebras were denoted $V_{\text{IIIb}}(E,J,\fl,\fh,\fg,\psi)$ with
underlying vector space $\RR(u,v) \oplus \RR(\be_+,\be_-) \oplus E
\oplus\fl \oplus\fh\oplus \fg$ with $\left<u,u\right> =
\left<v,v\right> = \left<\be_\pm,\be_\pm\right> = 0$,
$\left<u,v\right>=1=\left<\be_+,\be_-\right>$ and all $\oplus$
orthogonal.  The nonzero Lie 3-brackets are given by
\begin{equation}
  \label{eq:type-IIIb}
  \begin{aligned}[m]
    [u,\be_-,x] &= J x\\
    [u,x,y] &= \left<J x,y\right>\be_+\\
    [\be_-,x,y] &= - \left<Jx,y\right> v\\
    [\be_-,h_1,h_2] &= [h_1,h_2]_{\fh}\\
    [h_1,h_2,h_3] &= -\left<[h_1,h_2]_{\fh},h_3\right> \be_+
  \end{aligned}\qquad
  \begin{aligned}[m]
    [u,g_1,g_2] &= [\psi g_1,g_2]_{\fg}\\
    [\be_-,g_1,g_2] &= [g_1,g_2]_{\fg}\\
    [g_1,g_2,g_3] &= - \left<[g_1,g_2]_{\fg},g_3\right> \be_+ - \left<[\psi g_1,g_2]_{\fg}, g_3\right> v\\
    [u,\ell_1,\ell_2] &= [\ell_1,\ell_2]_{\fl}\\
    [\ell_1,\ell_2,\ell_3] &= -\left<[\ell_1,\ell_2]_{\fl},\ell_3\right> v,
  \end{aligned}
\end{equation}
where $x,y\in E$, $h,h_i\in\fh$, $g_i\in\fg$ and $\ell_i\in\fl$.

To see that this family of $3$-algebras is of the type \eqref{eq:2pIndec} it is enough to identify
\begin{equation}
  u_1 \leftrightarrow u  \qquad v_1 \leftrightarrow v \qquad u_2 \leftrightarrow e_- \qquad v_2 \leftrightarrow e_+
\end{equation}
as well as
\begin{equation}
  W_0 \leftrightarrow E \qquad\text{and}\qquad \bigoplus_{\alpha =1}^N  W_\alpha  \leftrightarrow  \fl \oplus\fh\oplus \fg,
\end{equation}
where the last identification is not only as vector spaces but also as Lie algebras, and set
\begin{equation} \label{eq:2pkappas}
  \begin{aligned}[m]
  \kappa_1|_{\fh} &= 0\\
  \kappa_1|_{\fl} &= 1 \\
  \kappa_1|_{\fg_{\alpha}} &= \psi_{\alpha}
  \end{aligned}\qquad\qquad
  \begin{aligned}[m]
    \kappa_2|_{\fh} &= 1\\
  \kappa_2|_{\fl} &= 0\\
  \kappa_2|_{\fg_{\alpha}} &= 1,
  \end{aligned}
\end{equation}
to obtain the map between the two families.  As shown in \cite{2p3Lie}
there are 9 different types of such 3-Lie algebras, 
depending on which of the four ingredients $(E,J)$, $\fl$, $\fh$ or
$(\fg,\psi)$ are present. 

The next case is that of index $r=3$, where there are up to 3
nonvanishing $J_{ij}$ and one $K_{123} := K$, while $L_{ijkl}$ is
still not present.  Indecomposability of $V$ forces $\dim E_0 \leq 1$.
As a vector space, $V$ splits up as
\begin{equation}
  V = \bigoplus_{i=1}^3 \left(\RR u_i \oplus \RR v_i\right) \oplus
  \bigoplus_{\alpha =1}^N W_\alpha \oplus \bigoplus_{\pi=1}^M E_\pi
  \oplus E_0,
\end{equation}
where all $\oplus$ are orthogonal except the second one, $W_\alpha$,
$E_0$ and $E_\pi$ are positive-definite subspaces with $\dim E_0 \leq
1$, $E_\pi$ being two-dimensional, and where $\left<u_i,u_j\right> =
\left<v_i,v_j\right> = 0$ and $\left<u_i,v_j\right> = \delta_{ij}$.
The 3-brackets are given by
\begin{equation}\label{eq:3pIndec}
  \begin{aligned}[m]
    [u_1,u_2,u_3] &= K \\
    [u_i,u_j,x_0] &= - \sum_{k=1}^r \left<K_{ijk},x_0\right> v_k\\
    [u_i,u_j,x_\pi] &= J^\pi_{ij} H_\pi x_\pi\\
    [u_i,x_\pi,y_\pi] &= - \sum_{j=1}^r \left<x_\pi,J^\pi_{ij} H_\pi y_\pi\right> v_j\\
    [u_i,x_\alpha,y_\alpha] &= \kappa_i^\alpha [x_\alpha,y_\alpha]\\
    [x_\alpha,y_\alpha,z_\alpha] &= - \left<[x_\alpha,y_\alpha],z_\alpha\right> \sum_{i=1}^r \kappa_i^\alpha v_i,
  \end{aligned}
\end{equation}
for all $x_0 \in E_0$, $x_\pi,y_\pi \in E_\pi$ and $x_\alpha, y_\alpha, z_\alpha \in W_\alpha$, and where $J^\pi_{ij} = \eta^\pi_i
\zeta^\pi_j - \eta^\pi_j \zeta^\pi_i$ and $H_\pi$ a complex structure on each 2-plane $E_\pi$.


Finally, let us remark that the family of admissible 3-Lie algebras found in \cite{Ho:2009nk} are included in
Theorem~\ref{thm:main}.  In that paper, a family of solutions to equations \eqref{eq:FI-V-pre} was found by setting each of the
Lie algebra structures $[-,-]_i$ to be nonzero in orthogonal subspaces of $W$.  This corresponds, in the language of this paper,
to the particular case of allowing precisely one $\kappa_i^\alpha$ to be nonvanishing in each $W_{\alpha}$.

Notice that, as shown in \eqref{eq:2pkappas}, already in \cite{2p3Lie} there are examples of admissible 3-Lie algebras of index
$2$ which are not of this form as both $\kappa_1$ and $\kappa_2$ might be nonvanishing in the $\fg_{\alpha}$ factors.

To solve the rest of the equations, two ansätze are proposed in \cite{Ho:2009nk}:
\begin{itemize}
\item the trivial solution with nonvanishing $J$,
  i.e. $\kappa_i^\alpha = 0$, $K_{ijk}=0$ for all $i,j,k = 1,...,r$ and
  for all $\alpha$; and

\item precisely one $\kappa_i^\alpha = 1$ for each $\alpha$ (and
  include those $W_{\alpha}$'s where all $\kappa$'s are zero in $W_0$)
  and one $J_{ij} := J \neq 0$ assumed to be an outer
  derivation of the reference Lie algebra defined on $W$.
\end{itemize}

As pointed out in that paper, $L_{ijkl}$ is not constrained by the fundamental identity, so it can in principle take any value,
whereas the ansatz provided for $K_{ijk}$ is given in terms of solutions of an equation equivalent to \eqref{eq:C8}.  In the
lagrangians considered, both $L_{ijkl}$ and $K_{ijk}$ are set to zero.

One thing to notice is that in all these theories there is certain redundancy concerning the index of the 3-Lie algebra.  If the
indices in the nonvanishing structures $\kappa_i^\alpha$, $J_{ij}$, $K_{ijk}$ and $L_{ijkl}$ involve only numbers from 1 to
$r_0$, then any 3-Lie algebra with such nonvanishing structures and index $r \geq r_0$ gives rise to the equivalent theories.

In this light, in the first ansatz considered, one can always define the non vanishing $J$ to be $J_{12}$ and then the
corresponding theory will be equivalent to one associated to the index-2 3-Lie algebras considered in \cite{2p3Lie}.

In the second case, the fact that $J$ is an outer derivation implies that it must live on the abelian part of $W$ as a Lie
algebra, since the semisimple part does not possess outer derivations.  This coincides with what was shown above, i.e., that
$J|_{W_{\alpha}} = 0$ for each $\alpha$.  Notice that each Lie algebra $[-,-]_i$ identically vanishes in $W_0$, therefore the
structure constants of the 3-Lie algebra do not mix $J$ and $[-,-]_i$.  The theories in \cite{Ho:2009nk} corresponding to this
ansatz also have $K_{ijk}=0$, whence again they are equivalent to the theory corresponding to the index-2 3-Lie algebra which was
denoted $V(E,J,\fh)$ in \cite{2p3Lie}.

\section{Bagger--Lambert lagrangians}
\label{sec:lagrangians}

In this section we will consider the physical properties of the Bagger--Lambert theory based on the most general kind of
admissible metric 3-Lie algebra, as described in Theorem~\ref{thm:main}.

In particular we will investigate the structure of the expansion of the corresponding Bagger--Lambert lagrangians around a vacuum
wherein the scalars in half of the null directions of the 3-Lie algebra take the constant values implied by the equations of
motion for the scalars in the remaining null directions, spanning the maximally isotropic centre.  This technique was also used in
\cite{Ho:2009nk} and is somewhat reminiscent of the novel Higgs mechanism that was first introduced by Mukhi and Papageorgakis
\cite{MukhiBL} in the context of the Bagger--Lambert theory based on the unique simple euclidean 3-Lie algebra $S_4$.  Recall that
precisely this strategy has already been employed in lorentzian signature in \cite{HIM-M2toD2rev}, for the class of
Bagger--Lambert theories found in \cite{GMRBL,BRGTV,HIM-M2toD2rev} based on the unique admissible lorentzian metric 3-Lie algebra
$W ( \fg )$, where it was first appreciated that this theory is perturbatively equivalent to $N=8$ super Yang--Mills theory on
$\RR^{1,2}$ with the euclidean semisimple gauge algebra $\fg$.  That is, there are no higher order corrections to the super
Yang--Mills lagrangian here, in contrast with the infinite set of corrections (suppressed by inverse powers of the gauge coupling)
found for the super Yang--Mills theory with $\fsu(2)$ gauge algebra arising from higgsing the Bagger--Lambert theory based on
$S_4$ in \cite{MukhiBL}.  This perturbative equivalence between the Bagger--Lambert theory based on $W( \fg )$ and maximally
supersymmetric Yang--Mills theory with euclidean gauge algebra $\fg$ has since been shown more rigorously in
\cite{BLSNoGhost,GomisSCFT,D2toD2}.

We will show that there exists a similar relation with $N=8$ super
Yang--Mills theory after expanding around the aforementioned
maximally supersymmetric vacuum the Bagger--Lambert theories based
on the more general physically admissible metric 3-Lie algebras we
have considered.  However, the gauge symmetry in the super
Yang--Mills theory is generally based on a particular indefinite
signature metric Lie algebra here that will be identified in terms
of the data appearing in Theorem~\ref{thm:main}.  The physical
properties of the these Bagger--Lambert theories will be shown to
describe particular combinations of decoupled super Yang-Mills
multiplets with euclidean gauge algebras and free maximally
supersymmetric massive vector multiplets.  We will identify
precisely how the physical moduli relate to the algebraic data in
Theorem~\ref{thm:main}.  We will also note how the theories
resulting from those finite-dimensional indefinite signature 3-Lie
algebras considered in \cite{Ho:2009nk} are recovered.

\subsection{Review of two gauge theories in indefinite signature}
\label{sec:reviewindef}

Before utilising the structural results of the previous section, let
us briefly review some general properties of the maximal $N=8$
supersymmetric Bagger--Lambert and Yang--Mills theories in
three-dimensional Minkowski space that will be of interest to us, when
the fields are valued in a vector space $V$ equipped with a metric of
indefinite signature.  We shall denote this inner product by
$\left<-,-\right>$ and take it to have general indefinite signature
$(r,r+n)$.  We can then define a null basis $e_A = (u_i, v_i, e_a)$ for
$V$, with $i=1,...,r$, $a=1,...,n$, such that $\left<u_i,v_j\right> =
\delta_{ij}$, $\left<u_i,u_j\right> = 0 = \left<v_i,v_j\right>$ and
$\left<e_a,e_b\right> = \delta_{ab}$.

For the sake of clarity in the forthcoming analysis, we will ignore
the fermions in these theories.  Needless to say that they both have a
canonical maximally supersymmetric completion and none of the
manipulations we will perform break any of the supersymmetries of the
theories.

\subsubsection{Bagger--Lambert theory}
\label{sec:reviewBL}

Let us begin by reviewing some details of the bosonic field content of
the Bagger--Lambert theory based on the 3-bracket $[-,-,-]$ defining a
metric 3-Lie algebra structure on $V$.  The components of the canonical
4-form for the metric 3-Lie algebra are $F_{ABCD} := \left< [ e_A ,
  e_B , e_C ] , e_D \right>$ (indices will be lowered and raised using
the metric $\left<e_A,e_B\right>$ and its inverse).  The bosonic
fields in the Bagger--Lambert theory have components $X_I^A$ and
$(\tilde{A}_\mu)^A{}_B = F^A{}_{BCD} A_\mu^{CD}$, corresponding
respectively to the scalars ($I=1,...,8$ in the vector of the
$\fso(8)$ R-symmetry) and the gauge field ($\mu =0,1,2$ on $\RR^{1,2}$
Minkowski space).  Although the supersymmetry transformations and
equations of motion can be expressed in terms of
$(\tilde{A}_\mu)^A_B$, the lagrangian requires it to be expressed as
above in terms of $A_\mu^{AB}$.

The bosonic part of the Bagger--Lambert lagrangian is given by
\begin{equation}
\label{eq:BLLag}
  \eL = -\half \left< D_\mu X_I , D^\mu X_I \right> + \eV (X) + \eL_{\text{CS}}~,
\end{equation}
where the scalar potential is
\begin{equation}
\label{eq:BLV}
\eV (X) = -\tfrac1{12} \left< [X_I,X_J,X_K],[X_I,X_J,X_K]\right>~,
\end{equation}
the Chern--Simons term is
\begin{equation}
\label{eq:BLCS}
\eL_{\text{CS}} = \half \left( A^{AB} \wedge d \tilde{A}_{AB} +
  \tfrac23 A^{AB} \wedge \tilde{A}_{AC} \wedge \tilde{A}^C{}_B
\right)~,
\end{equation}
and $D_\mu \phi^A = \partial_\mu \phi^A + (\tilde{A}_\mu)^A{}_B
\phi^B$ defines the action on any field $\phi$ valued in $V$ of the
derivative $D$ that is gauge-covariant with respect to
$\tilde{A}^A{}_B$.  The infinitesimal gauge transformations take the
form $\delta \phi^A = - {\tilde \Lambda}^A{}_B \phi^B$ and $\delta
(\tilde{A}_\mu)^A{}_B = \partial_\mu {\tilde \Lambda}^A{}_B +
(\tilde{A}_\mu)^A{}_C {\tilde \Lambda}^C{}_B - {\tilde \Lambda}^A{}_C
(\tilde{A}_\mu)^C{}_B$, where ${\tilde \Lambda}^A{}_B = F^A{}_{BCD}
\Lambda^{CD}$ in terms of an arbitrary skewsymmetric parameter
$\Lambda^{AB} = - \Lambda^{BA}$.

If we now assume that the indefinite signature metric 3-Lie algebra
above admits a maximally isotropic centre which we can take to be
spanned by the basis elements $v_i$ then the 4-form components $F_{v_i
  ABC}$ must all vanish identically.  There are two important physical
consequences of this assumption.  The first is that the covariant
derivative $D_\mu X_I^{u_i} = \partial_\mu X_I^{u_i}$.  The second is
that the tensors $F_{ABCD}$ and $F_{ABC}{}^G F_{DEFG} = F_{ABC}{}^g
F_{DEFg}$ which govern all the interactions in the Bagger--Lambert
lagrangian contain no legs in the $v_i$ directions.  Therefore the
components $A_\mu^{v_i A}$ of the gauge field do not appear at all in
the lagrangian while $X_I^{v_i}$ appear only in the free kinetic term
$- D_\mu X_I^{u_i} \partial^\mu X_I^{v_i} = - \partial_\mu
X_I^{u_i} \partial^\mu X_I^{v_i}$.  Thus $X_I^{v_i}$ can be integrated
out imposing that each $X_I^{u_i}$ be a harmonic function on
$\RR^{1,2}$ which must be a constant if the solution is to be
nonsingular. (We will assume this to be the case henceforth but
singular monopole-type solutions may also be worthy of investigation,
as in \cite{Verlinde:2008di}.)  It is perhaps just worth noting that,
in addition to setting $X_I^{u_i}$ constant, one must also set the
fermions in all the $u_i$ directions to zero which is necessary and
sufficient for the preservation of maximal supersymmetry here.

The upshot is that we now have $-\half \left< D_\mu X_I , D^\mu X_I \right> = -\half D_\mu X_I^a D^\mu X_I^a$ (with contraction
over only the euclidean directions of $V$) and each $X_I^{u_i}$ is taken to be constant in \eqref{eq:BLLag}.  Since both
$X_I^{v_i}$ and $A_\mu^{v_i A}$ are now absent, it will be more economical to define $X_I^i := X_I^{u_i}$ and $A_\mu^{ia} :=
A_\mu^{u_i a}$ henceforth.


\subsubsection{Super Yang--Mills theory}
\label{sec:reviewSYM}

Let us now perform an analogous review for $N=8$ super Yang--Mills theory, with gauge symmetry based on the Lie bracket $[-,-]$ defining a metric Lie algebra structure $\fg$ on $V$.  The components of the canonical 3-form on $\fg$ are $f_{ABC} := \left< [ e_A , e_B ] , e_C \right>$.  The bosonic fields in the theory consist of a gauge field $A_\mu^A$ and seven scalar fields $X_I^A$ (where now $I=1,...,7$ in the vector of the $\fso(7)$ R-symmetry) with all fields taking values in $V$.  The field strength for the gauge field takes the canonical form $F_{\mu\nu} = [ D_\mu , D_\nu ] = \partial_\mu A_\nu - \partial_\nu A_\mu + [ A_\mu , A_\nu ]$ in terms of the gauge-covariant derivative $D_\mu = \partial_\mu + [ A_\mu , -]$.  This theory is not scale-invariant and has a dimensionful coupling constant $\kappa$.

The bosonic part of the super Yang-Mills lagrangian is given by
\begin{equation}
\label{eq:SYM}
{\eL}^{SYM} ( A^A , X_I^A , \kappa | \fg ) = - \half \left< D_\mu X_I , D^\mu X_I \right> -\tfrac{\kappa^2}{4} \left< [ X_I , X_J ], [ X_I , X_J ] \right> - \tfrac{1}{4\kappa^2} \left< F_{\mu\nu} , F^{\mu\nu} \right>~.
\end{equation}
Noting explicitly the dependence on the data on the left hand side will be useful when we come to consider super Yang-Mills theories with a much more elaborate gauge structure.

Assuming now that $\fg$ admits a maximally isotropic centre, again spanned by the basis elements $v_i$, then the 3-form components
$f_{v_i AB}$ must all vanish identically.  This property implies $D X_I^{u_i} = d X_I^{u_i}$, $F^{u_i} = d A^{u_i}$ and that the
tensors $f_{ABC}$ and $f_{AB}{}^E f_{CDE} = f_{AB}{}^e f_{CDe}$ which govern all the interactions contain no legs in the $v_i$
directions.  Therefore $X_I^{v_i}$ and $A^{v_i}$ only appear linearly in their respective free kinetic terms, allowing them to be
integrated out imposing that $X_I^{u_i}$ is constant and $A^{u_i}$ is exact.  Setting the fermions in all the $u_i$ directions to
zero again ensures the preservation of maximal supersymmetry.

The resulting structure is that all the inner products using $\left< e_A , e_B \right>$ in \eqref{eq:SYM} are to be replaced with
$\left< e_a , e_b \right>$ while all $X_I^{u_i}$ are to be taken constant and $A^{u_i} = d \phi^{u_i}$, for some functions
$\phi^{u_i}$.  With both $X_I^{v_i}$ and $A^{v_i}$ now absent, it will be convenient to define $X_I^i := X_I^{u_i}$ and $\phi^{i}
:= \phi^{u_i}$ henceforth.

Let us close this review by looking in a bit more detail at the physical properties of a particular example of a super Yang--Mills
theory in indefinite signature with maximally isotropic centre, whose relevance will become clear in the forthcoming sections.
Four-dimensional Yang--Mills theories based on such gauge groups were studied in \cite{Tseytlin:1995yw}.  The gauge structure of
interest is based on the lorentzian metric Lie algebra defined by the double extension $\fd (E,\RR )$ of an even-dimensional
vector space $E$ with euclidean inner product.  Writing $V = \RR u \oplus \RR v \oplus E$ as a lorentzian vector space, the
nonvanishing Lie brackets of $\fd (E,\RR )$ are given by
\begin{equation}
  \label{eq:deE}
  [ u , x ] = J x~, \quad\quad [ x , y ] = - \left< x, Jy \right> v~,
\end{equation}
for all $x,y \in E$ where the skewsymmetric endomorphism $J \in \fso(E)$ is part of the data defining the double extension.  The
canonical 3-form for $\fd (E,\RR )$ therefore has only the components $f_{uab} = J_{ab}$ with respect to the euclidean basis $e_a$
on $E$.  It will be convenient to take $J$ to be nondegenerate and so the eigenvalues of $J^2$ will be negative-definite.

We shall define the positive number $\mu^2 := X_I^u X_I^u$ as the $SO(7)$-norm-squared of the constant 7-vector $X_I^u$ and the
projection operator $P_{IJ}^u := \delta_{IJ} - \mu^{-2} \, X_I^u X_J^u$ onto the hyperplane $\RR^6 \subset \RR^7$ orthogonal to
$X_I^u$.  It will also be convenient to define $x^a := X_I^u X_I^a$ as the projection of the seventh super Yang--Mills scalar
field along $X_I^u$ and ${\mathcal{D}} \Phi := d \Phi - d \phi^u \wedge J\Phi$ where $\Phi$ can be any $p$-form on $\RR^{1,2}$
taking values in $E$.  In terms of this data, the super Yang--Mills lagrangian ${\eL}^{SYM} ( ( d \phi^u , A^a ) , ( X_I^u ,
X_I^a ) , \kappa | \fd (E,\RR ) )$ can be more succinctly expressed as
\begin{multline}
  \label{eq:SYMde}
  - \half P_{IJ}^u {\mathcal{D}}_\mu X_I^a {\mathcal{D}}^\mu X_J^a +\tfrac{\kappa^2 \mu^2}{2} ( J^2 )_{ab} P_{IJ}^u X_I^a X_J^b
  - \tfrac{1}{4 \kappa^2} ( 2\,{\mathcal{D}}_{[ \mu} A_{\nu ]}^a ) ( 2\,{\mathcal{D}}^{[ \mu} A^{\nu ]\, a} )\\
  - \tfrac{1}{2 \mu^2} \left( {\mathcal{D}}_\mu x^a + \mu^2 J^{ab} A_\mu^b \right) \left( {\mathcal{D}}^\mu x^a + \mu^2 J^{ac}
    A^{\mu \, c} \right)~.
\end{multline}
From the first line we see that the six scalar fields $P_{IJ}^u X_J^a$ are massive with mass-squared given by the eigenvalues of the matrix $- \kappa^2 \mu^2 ( J^2 )_{ab}$.  All the fields couple to $d \phi^u$ through the covariant derivative ${\mathcal{D}}$, but the second line shows that only the seventh scalar $x^a$ couples to the gauge field $A^a$.  However, the gauge symmetry of \eqref{eq:SYMde} under the transformations $\delta A^a = {\mathcal{D}} \lambda^a$ and $\delta x^a = -\mu^2 J^{ab} \lambda^b$, for any parameter $\lambda^a \in E$, shows that $x^a$ is in fact pure gauge and can be removed in \eqref{eq:SYMde} by fixing $\lambda^a = \mu^{-2} ( J^{-1} )^{ab} x^b$.  The remaining gauge symmetry of \eqref{eq:SYMde} is generated by the transformations $\delta \phi^u = \alpha$ and $\delta \Phi = \alpha \, J \Phi$ for all fields $\Phi \in E$, where $\alpha$ is an arbitrary scalar parameter.  This is obvious since ${\mathcal{D}} = {\mbox{exp}}( \phi^u J ) d {\mbox{exp}}( - \phi^u J )$ and therefore, one can take ${\mathcal{D}} = d$ in \eqref{eq:SYMde} by fixing $\alpha = - \phi^u$.

Thus, in the gauge defined above, the lagrangian ${\eL}^{SYM} ( ( d \phi^u , A^a ) , ( X_I^u , X_I^a ) , \kappa | \fd (E,\RR ) )$
becomes simply
\begin{equation}
  \label{eq:SYMdegf}
  - \half P_{IJ}^u \partial_\mu X_I^a \partial^\mu X_J^a +\tfrac{\kappa^2 \mu^2}{2} ( J^2 )_{ab} P_{IJ}^u X_I^a X_J^b  -
  \tfrac{1}{4 \kappa^2} ( 2\, \partial_{[ \mu} A_{\nu ]}^a ) ( 2\, \partial^{[ \mu} A^{\nu ]\, a} ) + \tfrac{\mu^2}{2} ( J^2
  )_{ab} A_\mu^a A^{\mu\, b}~,
\end{equation}
describing ${\mbox{dim}}\, E$ decoupled free abelian $N=8$
supersymmetric massive vector multiplets, each of which contains
bosonic fields given by the respective gauge field $\tfrac{1}{\kappa}
\, A_\mu^a$ plus six scalars $P_{IJ}^u X_I^a$, all with the same
mass-squared equal to the respective eigenvalue of $- \kappa^2 \mu^2 (
J^2 )_{ab}$.

It is worth pointing out that one can also obtain precisely the theory above from a particular truncation of an $N=8$ super
Yang--Mills theory with euclidean semisimple Lie algebra $\fg$.  If one introduces a projection operator $P_{IJ}$ onto a
hyperplane $\RR^6 \subset \RR^7$ then one can rewrite the seven scalar fields in this euclidean theory in terms of the six
projected fields $P_{IJ} X_J^a$ living on the hyperplane plus the single scalar $y^a$ in the complementary direction.  Unlike in
the lorentzian theory above however, this seventh scalar is not pure gauge.  Indeed, if we expand the super Yang--Mills lagrangian
\eqref{eq:SYM} for this euclidean theory around a vacuum where $y^a$ is constant then this constant appears as a physical modulus
of the effective field theory, namely it gives rise to mass terms for the gauge field $A^a$ and the six projected scalars $P_{IJ}
X_J^a$.  If one then truncates the effective field theory to the Coulomb branch, such that the dynamical fields $A$ and $P_{IJ}
X_J$ take values in a Cartan subalgebra $\ft < \fg$ (while the constant vacuum expectation value $y \in \fg$), then the lagrangian
takes precisely the form \eqref{eq:SYMdegf} after making the following identifications.  First one must take $E = \ft$ whereby the
gauge field $A^a$ and coupling $\kappa$ are the the same for both theories.  Second one must identify the six-dimensional
hyperplanes occupied by the scalars $X_I^a$ in both theories such that $P_{IJ}^u$ in \eqref{eq:SYMdegf} is identified with
$P_{IJ}$ here.  Finally, the mass matrix for the euclidean theory is $- \kappa^2 [( \ad_y )^2 ]_{ab}$ which must be
identified with $- \kappa^2 \mu^2 ( J^2 )_{ab}$ in \eqref{eq:SYMdegf}.  This last identification requires some words of
explanation.  We have defined $\ad_y \Phi := [ y , \Phi ]$ for all $\Phi \in \fg$, where $[-,-]$ denotes the Lie bracket
on $\fg$.  Since we have truncated the dynamical fields to the Cartan subalgebra $\ft$, only the corresponding legs of
$(\ad_y)^2$ contribute to the mass matrix.  However, clearly $y$ must not also be contained in $\ft$ or else the resulting
mass matrix would vanish identically.  Indeed, without loss of generality, one can take $y$ to live in the orthogonal complement
$\ft^\perp \subset \fg$ since it is only these components which contribute to the mass matrix.  Thus, although $( \ad_y
)^2$ can be nonvanishing on $\ft$, $\ad_y$ cannot.  Thus we cannot go further and equate $\ad_y$ with $\mu J$, even
though their squares agree on $\ft$.  To summarise all this more succinctly, after the aforementioned gauge-fixing of the lorentzian
theory and truncation of the euclidean theory, we have shown that
\begin{equation}
\label{eq:SYMdehiggs}
{\eL}^{SYM} \left( \left( d \phi^u , A |_E \right) , \left( X_I^u , P_{IJ}^u X_J |_E , x|_E \right) , \kappa | \fd (E,\RR ) \right) = {\eL}^{SYM} \left( A|_{E}  , \left( P_{IJ} X_J |_{E} , y|_{E^\perp} \right) , \kappa | \fg \right) ~,
\end{equation}
where $E = \ft$, $y \in \ft^\perp \subset \fg$ is constant and $( \ad_y )^2 = \mu^2 J^2$ on $\ft$.  Of course, it is not
obvious that one can always solve this last equation for $y$ in terms of a given $\mu$ and $J$ nor indeed whether this restricts
ones choice of $\fg$.  However, it is the particular case of $\dim E = 2$ that will be of interest to us in the context of
the Bagger--Lambert theory in \ref{sec:L-E} where we shall describe a nontrivial solution for any rank-2 semisimple Lie algebra
$\fg$.  Obvious generalisations of this solution give strong evidence that the equation can in fact always be solved.


\subsection{Bagger--Lambert theory for admissible metric 3-Lie algebras}
\label{sec:admissiblelagrangians}

We will now substitute the data appearing in Theorem~\ref{thm:main} into the bosonic part of the Bagger--Lambert lagrangian \eqref{eq:BLLag}, that is after having integrated out $X_I^{v_i}$ to set all $X_I^i := X_I^{u_i}$ constant.

Since we will be dealing with components of the various tensors appearing in Theorem~\ref{thm:main}, we need to introduce some index notation for components of the euclidean subspace $\bigoplus_{\alpha =1}^N W_\alpha \oplus \bigoplus_{\pi=1}^M E_\pi \oplus E_0$.  To this end we partition the basis $e_a = ( e_{a_\alpha} , e_{a_\pi} , e_{a_0} )$ on the euclidean part of the algebra, where subscripts denote a basis for the respective euclidean subspaces.  For example, $a_\alpha = 1,..., {\mathrm{dim}} \, W_\alpha$ whose range can thus be different for each $\alpha$.  Similarly $a_0 = 1,..., {\mathrm{dim}} \, E_0$, while $a_\pi = 1,2$ for each two-dimensional space $E_\pi$.  Since the decomposition $\bigoplus_{\alpha =1}^N W_\alpha \oplus \bigoplus_{\pi=1}^M E_\pi \oplus E_0$ is orthogonal with respect to the euclidean metric $\left< e_a , e_b \right> = \delta_{ab}$, we can take only the components $\left< e_{a_\alpha} , e_{b_\alpha} \right> = \delta_{a_\alpha b_\alpha}$, $\left< e_{a_\pi} , e_{b_\pi} \right> = \delta_{a_\pi b_\pi}$ and $\left< e_{a_0} , e_{b_0} \right> = \delta_{a_0 b_0}$ to be nonvanishing.  Since these are all just unit metrics on the various euclidean factors then we will not need to be careful about raising and lowering repeated indices, which are to be contracted over the index range of a fixed value of $\alpha$, $\pi$ or $0$.  Summations of the labels $\alpha$ and $\pi$ will be made explicit.

In terms of this notation, we may write the data from Theorem~\ref{thm:main} in terms of the following nonvanishing components of the canonical 4-form $F_{ABCD}$ of the algebra
\begin{equation}
\label{eq:4formcomponents}
\begin{aligned}
F_{u_i a_\alpha b_\alpha c_\alpha} &= \kappa_i^\alpha f_{a_\alpha b_\alpha c_\alpha} \\
F_{u_i u_j a_\pi b_\pi} &= \left( \eta_i^\pi \zeta_j^\pi - \eta_j^\pi \zeta_i^\pi \right) \epsilon_{a_\pi b_\pi} \\
F_{u_i u_j u_k a_0} &= K_{ijk a_0} \\
F_{u_i u_j u_k u_l} &= L_{ijkl}~,
\end{aligned}
\end{equation}
where $f_{a_\alpha b_\alpha c_\alpha}$ denotes the canonical 3-form
for the simple metric Lie algebra structure $\fg_\alpha$ on $W_\alpha$
and we have used the fact that the 2x2 matrix $H_\pi$ has only
components $\epsilon_{a_\pi b_\pi} = -\epsilon_{b_\pi a_\pi}$,
with $\epsilon_{12} =-1$, on each 2-plane $E_\pi$.

A final point of notational convenience will be to define $Y^{AB} :=
X_I^A X_I^B$ and the projection $X_I^\xi := \xi_i X_I^i$ for any $\xi
\in \RR^r$.  Combining these definitions allows us to write certain
projections which often appear in the lagrangian like $Y^{\xi
  \varsigma} := X_I^\xi X_I^\varsigma$ and $Y^{\xi a} := X_I^\xi
X_I^a$ for any $\xi , \varsigma \in \RR^r$.  It will sometimes be
useful to write $Y^{\xi \xi} \equiv \| X^\xi \|^2 \geq
0$ where $\| X^\xi \|$ denotes the $SO(8)$-norm of the
vector $X_I^\xi$.  A similar shorthand will be adopted for projections
of components of the gauge field, so that $A_\mu^{\xi \varsigma} :=
\xi_i \varsigma_j A_\mu^{ij}$ and $A_\mu^{\xi a} := \xi_i A_\mu^{i
  a}$.

It will be useful to note that the euclidean components of the
covariant derivative $D_\mu X_I^A = \partial_\mu X_I^A +
(\tilde{A}_\mu)^A{}_B X_I^B$ from section~\ref{sec:reviewBL} can be
written
\begin{equation}
\label{eq:covder}
\begin{aligned}
  D_\mu X_I^{a_\alpha} &= \partial_\mu X_I^{a_\alpha} - \kappa_i^\alpha f^{a_\alpha b_\alpha c_\alpha} \left( 2\, A_\mu^{i b_\alpha} X_I^{c_\alpha} + A_\mu^{b_\alpha c_\alpha} X_I^i \right)  \\
  &=: \eD_\mu X_I^{a_\alpha} - 2\, B_\mu^{a_\alpha} X_I^{\kappa^\alpha} \\
  D_\mu X_I^{a_\pi} &= \partial_\mu X_I^{a_\pi} + 2\, \eta_i^\pi \zeta_j^\pi \epsilon^{a_\pi b_\pi} \left( A_\mu^{ij} X_I^{b_\pi} - A_\mu^{i b_\pi} X_I^j + A_\mu^{j b_\pi} X_I^i \right)  \\
  &= \partial_\mu X_I^{a_\pi} + 2\, \epsilon^{a_\pi b_\pi} \left( A_\mu^{\eta^\pi \zeta^\pi} X_I^{b_\pi} - A_\mu^{\eta^\pi b_\pi} X_I^{\zeta^\pi} + A_\mu^{\zeta^\pi b_\pi} X_I^{\eta^\pi} \right)  \\
  D_\mu X_I^{a_0} &= \partial_\mu X_I^{a_0} - K_{ijk}{}^{a_0} A_\mu^{ij} X_I^k~.
\end{aligned}
\end{equation}
The second line defines two new quantities on each $W_\alpha$, namely
$B_\mu^{a_\alpha} := \half f^{a_\alpha b_\alpha c_\alpha}
A_\mu^{b_\alpha c_\alpha}$ and the covariant derivative $\eD_\mu
X_I^{a_\alpha} := \partial_\mu X_I^{a_\alpha} -2\, f^{a_\alpha
  b_\alpha c_\alpha} \kappa_i^\alpha A_\mu^{i b_\alpha}
X_I^{c_\alpha}$.  The latter object is just the canonical covariant
derivative with respect to the projected gauge field
$\eA_\mu^{a_\alpha} := -2\, A_\mu^{\kappa^\alpha a_\alpha}$ on each
$W_\alpha$.  The associated field strength $\eF_{\mu\nu} = [ \eD_\mu ,
\eD_\nu ]$ has components
\begin{equation}
  \eF^{a_\alpha} = -2\, \kappa_i^\alpha \left( d A^{i a_\alpha} - \kappa_j^\alpha f^{a_\alpha b_\alpha c_\alpha} A^{i b_\alpha} \wedge A^{j c_\alpha} \right)~.
\end{equation}

Although somewhat involved, the nomenclature above will help us
understand more clearly the structure of the Bagger--Lambert
lagrangian.  Let us consider now the contributions to \eqref{eq:BLLag}
coming from the scalar kinetic terms, the sextic potential and the
Chern--Simons term in turn.

The kinetic terms for the scalar fields give
\begin{equation}
  \label{eq:BLkin-pre}
  -\half \left< D_\mu X_I , D^\mu X_I \right> = -\half \sum_{\alpha =1}^N D_\mu X_I^{a_\alpha} D^\mu X_I^{a_\alpha} -\half
  \sum_{\pi =1}^M D_\mu X_I^{a_\pi} D^\mu X_I^{a_\pi} -\half D_\mu X_I^{a_0} D^\mu X_I^{a_0}
\end{equation}
which expands to
\begin{multline}
  \label{eq:BLkin}
\sum_{\alpha =1}^N \left\{ -\half \eD_\mu X_I^{a_\alpha} \eD^\mu X_I^{a_\alpha} + 2\, X_I^{\kappa^\alpha} B_\mu^{a_\alpha}
  \eD^\mu X_I^{a_\alpha} -2\, Y^{\kappa^\alpha \kappa^\alpha} B_\mu^{a_\alpha} B^{\mu \, a_\alpha} \right\} \\
+  \sum_{\pi =1}^M \left\{ -\half \partial_\mu X_I^{a_\pi} \partial^\mu X_I^{a_\pi} -2\, \partial^\mu X_I^{a_\pi} \epsilon^{a_\pi
 b_\pi} \left( A_\mu^{\eta^\pi \zeta^\pi} X_I^{b_\pi} - A_\mu^{\eta^\pi b_\pi} X_I^{\zeta^\pi} + A_\mu^{\zeta^\pi b_\pi}
    X_I^{\eta^\pi} \right) \right. \\
\left. -2\, \left( A_\mu^{\eta^\pi \zeta^\pi} X_I^{a_\pi} - A_\mu^{\eta^\pi a_\pi} X_I^{\zeta^\pi} + A_\mu^{\zeta^\pi a_\pi}
    X_I^{\eta^\pi} \right) \left( A^{\mu \, \eta^\pi \zeta^\pi} X_I^{a_\pi} - A^{\mu \, \eta^\pi a_\pi} X_I^{\zeta^\pi} + A^{\mu \, \zeta^\pi
    a_\pi} X_I^{\eta^\pi} \right) \right\} \\
-\half \partial_\mu X_I^{a_0} \partial^\mu X_I^{a_0} + K_{ijk}{}^{a_0} A_\mu^{ij} \partial^\mu Y^{k a_0} - \half K_{ijk a_0}
K_{lmn a_0} Y^{kl} A_\mu^{ij} A^{\mu \, mn}~. 
\end{multline}

The scalar potential can be written $\eV (X) = \eV^W (X) + \eV^E (X) + \eV^{E_0} (X)$ where
\begin{equation}
\label{eq:BLpot}
\begin{aligned}[m]
  \eV^W (X) &= -\tfrac{1}{4} \sum_{\alpha =1}^N f^{a_\alpha b_\alpha e_\alpha} f^{c_\alpha d_\alpha e_\alpha} \left( Y^{\kappa^\alpha \kappa^\alpha} Y^{a_\alpha c_\alpha} - Y^{\kappa^\alpha a_\alpha} Y^{\kappa^\alpha c_\alpha} \right) Y^{b_\alpha d_\alpha} \\
  \eV^E (X) &= -\tfrac{1}{2} \sum_{\pi =1}^M \left\{ Y^{a_\pi a_\pi} \left( Y^{\eta^\pi \eta^\pi} Y^{\zeta^\pi \zeta^\pi} - (
      Y^{\eta^\pi \zeta^\pi} )^2 \right) + 2 \, Y^{\eta^\pi a_\pi} Y^{\zeta^\pi a_\pi} Y^{\eta^\pi \zeta^\pi}\right. \\
    &\quad \left. - Y^{\eta^\pi a_\pi} Y^{\eta^\pi a_\pi} Y^{\zeta^\pi \zeta^\pi} - Y^{\zeta^\pi a_\pi} Y^{\zeta^\pi a_\pi} Y^{\eta^\pi \eta^\pi} \right\} \\
  \eV^{E_0} (X) &= -\tfrac{1}{12} K_{ijk a_0} K_{lmn a_0} Y^{il} Y^{jm} Y^{kn}~.
\end{aligned}
\end{equation}
Notice that $\eV^{E_0} (X)$ is constant and will be ignored henceforth.

And finally, the Chern--Simons term can be written $\eL_{\text{CS}} = \eL_{\text{CS}}^W + \eL_{\text{CS}}^E + \eL_{\text{CS}}^{E_0}$ where
\begin{equation}
\label{eq:BLCS2}
\begin{aligned}
\eL_{\text{CS}}^W &= -2\, \sum_{\alpha =1}^N B^{a_\alpha} \wedge \eF^{a_\alpha} \\
\eL_{\text{CS}}^E &= -4\, \sum_{\pi =1}^M \left\{ \epsilon^{a_\pi b_\pi} \, A^{\eta^\pi a_\pi} \wedge A^{\zeta^\pi b_\pi} +2 \, A^{\eta^\pi \zeta^\pi} \wedge A^{\eta^\pi a_\pi} \wedge A^{\zeta^\pi a_\pi} - \half \epsilon^{a_\pi b_\pi} A^{a_\pi b_\pi} \wedge d A^{\eta^\pi \zeta^\pi} \right\} \\
\eL_{\text{CS}}^{E_0} &= 2\, K_{ijk a_0} A^{ij} \wedge d A^{k a_0} - \tfrac{1}{3} K_{ikl a_0} K_{jmn a_0} A^{ij} \wedge A^{kl} \wedge A^{mn} + \half L_{ijkl} A^{ij} \wedge d A^{kl}~.
\end{aligned}
\end{equation}
These expressions are valid only up to total derivative terms that will be discarded.

Clearly there is a certain degree of factorisation for the
Bagger--Lambert lagrangian into separate terms living on the different
components of $\bigoplus_{\alpha =1}^N W_\alpha \oplus
\bigoplus_{\pi=1}^M E_\pi \oplus E_0$.  Indeed let us define
accordingly $\eL^W = -\half \sum_{\alpha =1}^N D_\mu X_I^{a_\alpha}
D^\mu X_I^{a_\alpha} + \eV^W (X) + \eL_{\text{CS}}^W$ and likewise for
$E$ and $E_0$.  This is mainly for notational convenience however and
one must be wary of the fact that $\eL^E$ and $\eL^{E_0}$ could have
some fields, namely components of $A^{ij}$, in common.

To relate the full lagrangian $\eL$ with a super Yang-Mills theory,
one has first to identify and integrate out those fields which are
auxiliary or appear linearly as Lagrange multipliers.  This will be
most easily done by considering $\eL^W$, $\eL^E$ and $\eL^{E_0}$ in
turn.

\subsubsection{$\eL^W$}
\label{sec:L-W}

The field $B^{a_\alpha}$ appears only algebraically as an auxiliary field in $\eL^W$.  Its equation of motion implies
\begin{equation}
\label{eq:eomB}
2 \, Y^{\kappa^\alpha \kappa^\alpha} B^{a_\alpha} = X_I^{\kappa^\alpha} \eD X_I^{a_\alpha} + {*\eF}^{a_\alpha}~,
\end{equation}
for each value of $\alpha$.  Substituting this back into $\eL^W$ then gives
\begin{equation}
\label{eq:intB}
-\half \sum_{\alpha =1}^N D_\mu X_I^{a_\alpha} D^\mu X_I^{a_\alpha} + \eL_{\text{CS}}^W = \sum_{\alpha =1}^N \left\{ - \half P_{IJ}^{\kappa^\alpha} \eD_\mu X_I^{a_\alpha} \eD^\mu X_J^{a_\alpha} - \tfrac{1}{4 Y^{\kappa^\alpha \kappa^\alpha}} \eF_{\mu\nu}^{a_\alpha} \eF^{\mu\nu\, a_\alpha} \right\}~,
\end{equation}
where, for each $\alpha$, $P_{IJ}^{\kappa^\alpha} := \delta_{IJ} -
\frac{X_I^{\kappa^\alpha} X_J^{\kappa^\alpha}}{Y^{\kappa^\alpha
    \kappa^\alpha}}$ is the projection operator onto the hyperplane
$\RR^7 \subset \RR^8$ which is orthogonal to the 8-vector
$X_I^{\kappa^\alpha}$ that $\kappa_i^\alpha$ projects the constant
$X_I^i$ onto.

Furthermore, in terms of the Lie bracket $[-,-]_\alpha$ on $\fg_\alpha$, the scalar potential can be written
\begin{equation}
\label{eq:potB}
\eV^W (X) = - \tfrac{1}{4} \sum_{\alpha =1}^N Y^{\kappa^\alpha \kappa^\alpha} \, P_{IK}^{\kappa^\alpha} P_{JL}^{\kappa^\alpha} \, [ X_I , X_J ]_\alpha^{a_\alpha} [ X_K , X_L ]_\alpha^{a_\alpha}~.
\end{equation}

In conclusion, we have shown that upon integrating out $B^{a_\alpha}$ one can identify
\begin{equation}
\label{eq:LalphaSYM}
\eL^W = \sum_{\alpha =1}^N {\eL}^{SYM} \left( \eA^{a_\alpha} , P_{IJ}^{\kappa^\alpha} X_J^{a_\alpha} , \| X^{\kappa^\alpha} \| | \fg_\alpha \right)~.
\end{equation}
The identification above with the lagrangian in \eqref{eq:SYM} has
revealed a rather intricate relation between the data
$\kappa_i^\alpha$ and $\fg_\alpha$ on $W_\alpha$ from
Theorem~\ref{thm:main} and the physical parameters in the super
Yang--Mills theory.  In particular, the coupling constant for the super
Yang--Mills theory on $W_\alpha$ corresponds to the $SO(8)$-norm of
$X_I^{\kappa^\alpha}$.  Moreover, the direction of
$X_I^{\kappa^\alpha}$ in $\RR^8$ determines which hyperplane the seven
scalar fields in the super Yang--Mills theory must occupy and thus may
be different on each $W_\alpha$.  The gauge symmetry is based on the
euclidean Lie algebra $\bigoplus_{\alpha =1}^N \fg_\alpha$.

The main point to emphasise is that it is the projections of the
individual $\kappa_i^\alpha$ onto the vacuum described by constant
$X_I^i$ (rather than the vacuum expectation values themselves) which
determine the physical moduli in the theory.  For example, take $N=1$
with only one simple Lie algebra structure $\fg = \fsu(n)$ on $W$.  The
lagrangian \eqref{eq:LalphaSYM} then describes precisely the
low-energy effective theory for $n$ coincident D2-branes in type IIA
string theory, irrespective of the index $r$ of the initial 3-Lie
algebra.  The only difference is that the coupling $\|
X^{\kappa} \|$, to be interpreted as the perimeter of the
M-theory circle, is realised as a different projection for different
values of $r$.

Thus, in general, we are assuming a suitably generic situation wherein
none of the projections $X_I^{\kappa^\alpha}$ vanish identically.  If
$X_I^{\kappa^\alpha} =0$ for a given value of $\alpha$ then the
$W_\alpha$ part of the scalar potential \eqref{eq:BLpot} vanishes
identically and the only occurrence of the corresponding
$B^{a_\alpha}$ is in the Chern--Simons term \eqref{eq:BLCS2}.  Thus,
for this particular value of $\alpha$, $B^{a_\alpha}$ has become a
Lagrange multiplier imposing $\eF^{a_\alpha} = 0$ and so
$\eA^{a_\alpha}$ is pure gauge.  The resulting lagrangian on this
$W_\alpha$ therefore describes a free $N=8$ supersymmetric theory for
the eight scalar fields $X_I^{a_\alpha}$.

\subsubsection{$\eL^E$}
\label{sec:L-E}

The field $\epsilon^{a_\pi b_\pi} A^{a_\pi b_\pi}$ appears
only linearly in one term in $\eL_{\text{CS}}^E$ and is therefore a
Lagrange multiplier imposing the constraint $A^{\eta^\pi
  \zeta^\pi} = d \phi^{\eta^\pi \zeta^\pi}$, for some some
scalar fields $\phi^{\eta^\pi \zeta^\pi}$, for each value of
$\pi$.  The number of distinct scalars $\phi^{\eta^\pi
  \zeta^\pi}$ will depend on the number of linearly independent
2-planes in $\RR^r$ which the collection of all $\eta^\pi \wedge
\zeta^\pi$ span for $\pi =1,...,M$.  Let us henceforth call this
number $k$, which is clearly bounded above by $\binom{r}{2}$.

Moreover, up to total derivatives, one has a choice of taking just one
of the two gauge fields $A^{\eta^\pi a_\pi}$ and $A^{\zeta^\pi
  a_\pi}$ to be auxiliary in $\eL^E$.  These are linearly independent
gauge fields by virtue of the fact that $\eta^\pi \wedge
\zeta^\pi$ span a 2-plane in $\RR^r$ for each value of
$\pi$.  Without loss of generality we can take $A^{\eta^\pi
  a_\pi}$ to be auxiliary and integrate it out in favour of
$A^{\zeta^\pi a_\pi}$.  After implementing the Lagrange multiplier
constraint above, one finds that the equation of motion of
$A^{\eta^\pi a_\pi}$ implies
\begin{multline}
\label{eq:eomAeta}
2\, Y^{\zeta^\pi \zeta^\pi} A^{\eta^\pi a_\pi} = -
\epsilon^{a_\pi b_\pi} \left\{ X_I^{\zeta^\pi} \left( d
    X_I^{b_\pi} +2\, \epsilon^{b_\pi c_\pi} \left( X_I^{c_\pi}
      d \phi^{\eta^\pi \zeta^\pi} + X_I^{\eta^\pi}
      A^{\zeta^\pi c_\pi}  \right) \right) \right. \\
\left. {} + 2\, {*\left( d
      A^{\zeta^\pi b_\pi} + 2\, \epsilon^{b_\pi c_\pi} d
      \phi^{\eta^\pi \zeta^\pi} \wedge A^{\zeta^\pi c_\pi}
    \right)} \right\}~.
\end{multline}
Substituting this back into $\eL^E$ then, following a rather lengthy but straightforward calculation, one finds that
\begin{equation}
\label{eq:intAeta}
\begin{aligned}
-\half \sum_{\pi =1}^M D_\mu X_I^{a_\pi} D^\mu X_I^{a_\pi} + \eL_{\text{CS}}^E =& - \half \sum_{\pi =1}^M P_{IJ}^{\zeta^\pi} \left( \partial_\mu X_I^{a_\pi} +2\, \epsilon^{a_\pi b_\pi} \left( X_I^{b_\pi} \partial_\mu \phi^{\eta^\pi \zeta^\pi} + X_I^{\eta^\pi} A_\mu^{\zeta^\pi b_\pi} \right) \right) \\
&\hspace*{.9in} \times \left( \partial^\mu X_J^{a_\pi} +2\, \epsilon^{a_\pi c_\pi} \left( X_J^{c_\pi} \partial^\mu \phi^{\eta^\pi \zeta^\pi} + X_J^{\eta^\pi} A^{\mu\, \zeta^\pi c_\pi} \right) \right) \\
& - \sum_{\pi =1}^M \tfrac{4}{Y^{\zeta^\pi \zeta^\pi}} \left( \partial_{[\mu} A_{\nu ]}^{\zeta^\pi a_\pi} + 2\, \epsilon^{a_\pi b_\pi} \partial_{[\mu} \phi^{\eta^\pi \zeta^\pi} A_{\nu ]}^{\zeta^\pi b_\pi} \right) \\
&\hspace*{.9in} \times \left( \partial^{\mu} A^{\nu\, \zeta^\pi a_\pi} + 2\, \epsilon^{a_\pi c_\pi} \partial^{\mu} \phi^{\eta^\pi \zeta^\pi} A^{\nu\, \zeta^\pi c_\pi} \right) ~,
\end{aligned}
\end{equation}
where, for each $\pi$, $P_{IJ}^{\zeta^\pi} := \delta_{IJ} -
\frac{X_I^{\zeta^\pi} X_J^{\zeta^\pi}}{Y^{\zeta^\pi
    \zeta^\pi}}$ projects onto the hyperplane $\RR^7 \subset \RR^8$
orthogonal to the 8-vector $X_I^{\zeta^\pi}$ which $\zeta_i^\pi$
projects the constant $X_I^i$ onto.

We have deliberately written \eqref{eq:intAeta} in a way that is
suggestive of a super Yang--Mills description for the fields on $E$
however, in contrast with the preceding analysis for $W$, the gauge
structure here is not quite so manifest.  To make it more transparent,
let us fix a particular value of $\pi$ and consider a 4-dimensional
lorentzian vector space of the form $\RR e_+ \oplus \RR e_- \oplus
E_\pi$, where the particular basis $( e_+ , e_- )$ for the two null
directions obeying $\left< e_+ , e_- \right> =1$ and $\left< e_\pm ,
  e_\pm \right> = 0 = \left< e_\pm , e_{a_\pi} \right>$ can of
course depend on the choice of $\pi$ (we will omit the $\pi$ label
here though).  If we take $E_\pi$ to be a euclidean 2-dimensional
abelian Lie algebra then we can define a lorentzian metric Lie algebra
structure on $\RR e_+ \oplus \RR e_- \oplus E_\pi$ given by the
double extension $\fd ( E_\pi , \RR )$.  The nonvanishing Lie
brackets of $\fd ( E_\pi , \RR )$ are
\begin{equation}
\label{eq:deE-2}
[ e_+ , e_{a_\pi} ] = - \epsilon_{a_\pi b_\pi} e_{b_\pi}~, \quad\quad [ e_{a_\pi} , e_{b_\pi} ] = - \epsilon_{a_\pi b_\pi} e_-~.
\end{equation}
This double extension is precisely the Nappi--Witten Lie algebra.

For each value of $\pi$ we can collect the following sets of scalars
${\sf X}_I^\pi := ( X_I^{\eta^\pi} , X_I^{\zeta^\pi} ,
X_I^{a_\pi} )$ and gauge fields ${\sf A}^\pi := ( 2\, d
\phi^{\eta^\pi \zeta^\pi} , 0 , -2\, A^{\zeta^\pi a_\pi} )$
into elements of the aforementioned vector space $\RR e_+ \oplus \RR
e_- \oplus E_\pi$.  The virtue of doing so being that if ${\sf D} = d
+ [ {\sf A},-]$, for each value of $\pi$, is the canonical
gauge-covariant derivative with respect to each $\fd ( E_\pi , \RR
)$ then $( {\sf D} {\sf X}_I )^{a_\pi} = d X_I^{a_\pi} +2\,
\epsilon^{a_\pi b_\pi} \left( X_I^{b_\pi} d \phi^{\eta^\pi
    \zeta^\pi} + X_I^{\eta^\pi} A^{\zeta^\pi b_\pi} \right)$
while the associated field strength ${\sf F}_{\mu\nu} = [ {\sf D}_\mu
, {\sf D_\nu} ]$ has ${\sf F}^{a_\pi} = -2\, \left( d A^{\zeta^\pi
    a_\pi} + 2\, \epsilon^{a_\pi b_\pi} d \phi^{\eta^\pi
    \zeta^\pi} \wedge A^{\zeta^\pi b_\pi} \right)$.  These are
exactly the components appearing in \eqref{eq:intAeta}!

Moreover, the scalar potential $\eV^{E} (X)$ can be written
\begin{equation}
\label{eq:potAeta}
\eV^E (X) = - \tfrac{1}{4} \sum_{\pi =1}^M Y^{\zeta^\pi \zeta^\pi} \, P_{IK}^{\zeta^\pi} P_{JL}^{\zeta^\pi} \, [ {\sf X}_I , {\sf X}_J ]^{a_\pi} [ {\sf X}_K , {\sf X}_L ]^{a_\pi}~,
\end{equation}
where $[-,-]$ denotes the Lie bracket on each $\fd ( E_\pi , \RR )$ factor.

Thus it might appear that $\eL^E$ is going to describe a super
Yang--Mills theory whose gauge algebra is $\bigoplus_{\pi =1}^M \fd
( E_\pi , \RR )$, which indeed has a maximally isotropic centre and
so is of the form noted in section~\ref{sec:reviewSYM}.  However,
this need not be the case in general since the functions
$\phi^{\eta^\pi \zeta^\pi}$ appearing in the $e_+$ direction of
each ${\sf A}^\pi$ must describe the same degree of freedom for
different values of $\pi$ precisely when the corresponding 2-planes
in $\RR^r$ spanned by $\eta^\pi \wedge \zeta^\pi$ are linearly
dependent.  Consequently we must identify the $( e_+ , e_- )$
directions in all those factors $\fd ( E_\pi , \RR )$ for which the
associated $\eta^\pi \wedge \zeta^\pi$ span the same 2-plane in
$\RR^r$.  It is not hard to see that, with respect to a general basis
on $\bigoplus_{\pi =1}^M E_\pi$, the resulting Lie algebra $\fk$
must take the form $\bigoplus_{{[\pi]} =1}^k \fd ( E_{[\pi]} , \RR )$ of
an orthogonal direct sum over the number of independent 2-planes $k$
spanned by $\eta^{[\pi]} \wedge \zeta^{[\pi]}$ of a set of $k$ double
extensions $\fd ( E_{[\pi]} , \RR )$ of even-dimensional vector spaces
$E_{[\pi]}$, where $\bigoplus_{\pi =1}^M E_\pi = \bigoplus_{{[\pi]}
  =1}^k E_{[\pi]}$.  That is each ${[\pi]}$ can be thought of as
encompassing an equivalence class of $\pi$ values for which the
corresponding 2-forms $\eta^\pi \wedge \zeta^\pi$ are all
proportional to each other.  The data for $\fk$ therefore corresponds
to a set of $k$ nondegenerate elements $J_{[\pi]} \in \fso( E_{[\pi]}
)$ where, for a given value of ${[\pi]}$, the relative eigenvalues of
$J_{[\pi]}$ are precisely the relative proportionality constants for
the linearly dependent 2-forms $\eta^\pi \wedge \zeta^\pi$ in the
equivalence class.  Clearly $\fk$ therefore has index $k$, dimension $2
\left( k + \left[ \tfrac{{\mathrm{dim}}\, W_0}{2} \right] \right)$ and
admits a maximally isotropic centre.

Putting all this together, we conclude that
\begin{equation}
  \label{eq:LbetaSYM}
  \eL^E = \sum_{{[\pi]} =1}^k {\eL}^{SYM} \left( {\sf A}^{{[\pi]}} ,
    P_{IJ}^{\zeta^{[\pi]}} {\sf X}_J^{{[\pi]}} , \| X^{\zeta^{[\pi]}}
    \| \middle | \fd ( E_{[\pi]} , \RR ) \right)~.
\end{equation}
One can check from \eqref{eq:BLpot} and \eqref{eq:intAeta} that the
contributions to the Bagger--Lambert lagrangian on $E$ coming from
different $E_\pi$ factors, but with $\pi$ values in the same
equivalence class ${[\pi]}$, are precisely accounted for in the
expression \eqref{eq:LbetaSYM} by the definition above of the elements
$J_{[\pi]}$ defining the double extensions.

The identification above again provides quite an intricate relation
between the data on $E_\pi$ from Theorem~\ref{thm:main} and the
physical super Yang--Mills parameters.  However, we know from
section~\ref{sec:reviewSYM} that the physical content of super
Yang--Mills theories whose gauge symmetry is based on a lorentzian Lie
algebra corresponding to a double extension is rather more simple,
being described in terms of free massive vector supermultiplets.  Let
us therefore apply this preceding analysis to the theory above.

The description above of the lagrangian on each factor $E_\pi$ has
involved projecting degrees of freedom onto the hyperplane $\RR^7
\subset \RR^8$ orthogonal to $X_I^{\zeta^\pi}$.  The natural analogy
here of the six-dimensional subspace occupied by the massive scalar
fields in section~\ref{sec:reviewSYM} is obtained by projecting
onto the subspace $\RR^6 \subset \RR^8$ which is orthogonal to the
plane in $\RR^8$ spanned by $X^{\eta^\pi} \wedge X^{\zeta^\pi}$,
i.e. the image in $\Lambda^2 \RR^8$ of the 2-form $\eta^\pi \wedge
\zeta^\pi$ under the map from $\RR^r \to \RR^8$ provided by the
vacuum expectation values $X_I^i$.  This projection operator can be
written
\begin{equation}
  \label{eq:2proj}
  P_{IJ}^{\eta^\pi \zeta^\pi} = \delta_{IJ} - X_I^{\eta^\pi} Q_J^{\eta^\pi} - X_I^{\zeta^\pi} Q_J^{\zeta^\pi}~,
\end{equation}
where
\begin{equation}
  \label{eq:dualX}
  \begin{aligned}
    Q_I^{\eta^\pi} &:= \frac{1}{(\Delta_{\eta^\pi \zeta^\pi} )^2} \left( Y^{\zeta^\pi \zeta^\pi} X_I^{\eta^\pi} - Y^{\eta^\pi \zeta^\pi} X_I^{\zeta^\pi} \right) \\
    Q_I^{\zeta^\pi} &:= \frac{1}{(\Delta_{\eta^\pi \zeta^\pi}
      )^2} \left( Y^{\eta^\pi \eta^\pi} X_I^{\zeta^\pi} -
      Y^{\eta^\pi \zeta^\pi} X_I^{\eta^\pi} \right)~,
  \end{aligned}
\end{equation}
and
\begin{equation}
  \label{eq:areamass}
  (\Delta_{\eta^\pi \zeta^\pi} )^2 := \| X^{\eta^\pi}
  \wedge X^{\zeta^\pi} \|^2 \equiv Y^{\eta^\pi \eta^\pi}
  Y^{\zeta^\pi \zeta^\pi} - ( Y^{\eta^\pi \zeta^\pi} )^2~.
\end{equation}
The quantities defined in \eqref{eq:dualX} are the dual elements to
$X_I^{\eta^\pi}$ and $X_I^{\zeta^\pi}$ such that $Q_I^{\eta^\pi}
X_I^{\eta^\pi} = 1 = Q_I^{\zeta^\pi} X_I^{\zeta^\pi}$ and
$Q_I^{\eta^\pi} X_I^{\zeta^\pi} = 0 = Q_I^{\zeta^\pi}
X_I^{\eta^\pi}$.  The expression \eqref{eq:areamass} identifies
$\Delta_{\eta^\pi \zeta^\pi}$ with the area in $\RR^8$ spanned by
$X^{\eta^\pi} \wedge X^{\zeta^\pi}$.  From these definitions, it
follows that $P_{IJ}^{\eta^\pi \zeta^\pi}$ in \eqref{eq:2proj}
indeed obeys $P_{IJ}^{\eta^\pi \zeta^\pi} = P_{JI}^{\eta^\pi
  \zeta^\pi}$, $P_{IK}^{\eta^\pi \zeta^\pi} P_{JK}^{\eta^\pi
  \zeta^\pi} = P_{IJ}^{\eta^\pi \zeta^\pi}$ and
$P_{IJ}^{\eta^\pi \zeta^\pi} X_J^{\eta^\pi} = 0 =
P_{IJ}^{\eta^\pi \zeta^\pi} X_J^{\zeta^\pi}$.

The scalar potential \eqref{eq:potAeta} on $E$ has a natural expression in terms of the objects defined in \eqref{eq:2proj} and
\eqref{eq:areamass} as
\begin{equation}
\label{eq:pot2proj}
\eV^E (X) = - \tfrac{1}{2} \sum_{\pi =1}^M (\Delta_{\eta^\pi \zeta^\pi} )^2 \, P_{IJ}^{\eta^\pi \zeta^\pi} X_I^{a_\pi} X_J^{a_\pi}~.
\end{equation}

Furthermore, using the identity
\begin{equation}
  \label{eq:2projid}
  P_{IJ}^{\eta^\pi \zeta^\pi} \equiv P_{IJ}^{\zeta^\pi} -
  \frac{(\Delta_{\eta^\pi \zeta^\pi} )^2}{Y^{\zeta^\pi
      \zeta^\pi}} Q_I^{\eta^\pi} Q_J^{\eta^\pi}~,
\end{equation}
allows one to reexpress the remaining terms
\begin{equation}
  -\half \sum_{\pi =1}^M D_\mu X_I^{a_\pi} D^\mu X_I^{a_\pi} + \eL_{\text{CS}}^E
\end{equation}
in \eqref{eq:intAeta} as
\begin{multline}
  \label{eq:intAeta2}
\sum_{\pi =1}^M -\half \, P_{IJ}^{\eta^\pi
      \zeta^\pi} {\mathcal{D}}_\mu X_I^{a_\pi} {\mathcal{D}}^\mu
    X_J^{a_\pi} - \tfrac{1}{Y^{\zeta^\pi \zeta^\pi}} \left( 2\,
      {\mathcal{D}}_{[\mu} A_{\nu ]}^{\zeta^\pi a_\pi} \right)
    \left( 2\, {\mathcal{D}}^{\mu} A^{\nu\, \zeta^\pi a_\pi} \right)
    \\
    - \half \sum_{\pi =1}^M \frac{Y^{\zeta^\pi
        \zeta^\pi}}{(\Delta_{\eta^\pi \zeta^\pi} )^2} \left(
      X_I^{\eta^\pi} P_{IJ}^{\zeta^\pi} {\mathcal{D}}_\mu
      X_J^{a_\pi} + 2\, \frac{(\Delta_{\eta^\pi \zeta^\pi}
        )^2}{Y^{\zeta^\pi \zeta^\pi}} \epsilon^{a_\pi b_\pi}
      A_\mu^{\zeta^\pi b_\pi} \right) \\
     \times \left( X_K^{\eta^\pi} P_{KL}^{\zeta^\pi}
      {\mathcal{D}}^\mu X_L^{a_\pi} + 2\, \frac{(\Delta_{\eta^\pi
          \zeta^\pi} )^2}{Y^{\zeta^\pi \zeta^\pi}}
      \epsilon^{a_\pi c_\pi} A^{\mu \, \zeta^\pi c_\pi}
    \right)~,
\end{multline}
where we have introduced the covariant derivative ${\mathcal{D}}
\Phi^{a_\pi} := d \Phi^{a_\pi} + 2\, \epsilon^{a_\pi b_\pi} \,
d \phi^{\eta^\pi \zeta^\pi} \wedge \Phi^{b_\pi}$ for any
differential form $\Phi^{a_\pi}$ on $\RR^{1,2}$ taking values in
$E_\pi$.  Similar to what we saw in section~\ref{sec:reviewSYM},
the six projected scalars $P_{IJ}^{\eta^\pi \zeta^\pi}
X_J^{a_\pi}$ in the first line of \eqref{eq:intAeta2} do not couple
to the gauge field $A^{\zeta^\pi a_\pi}$ on each
$E_\pi$.  Moreover, the remaining scalar in the second line of
\eqref{eq:intAeta2} can be eliminated from the lagrangian, for each
$E_\pi$, using the gauge symmetry under which $\delta A^{i a_\pi}
= {\mathcal{D}} \Lambda^{i a_\pi}$ for any parameter $\Lambda^{i
  a_\pi}$ to fix $\Lambda^{\zeta^\pi a_\pi} = - \half
\frac{Y^{\zeta^\pi \zeta^\pi}}{(\Delta_{\eta^\pi \zeta^\pi}
  )^2} \epsilon^{a_\pi b_\pi} X_I^{\eta^\pi}
P_{IJ}^{\zeta^\pi} X_J^{b_\pi}$.  There is a remaining gauge
symmetry under which $\delta \phi^{\eta^\pi \zeta^\pi} =
\Lambda^{\eta^\pi \zeta^\pi}$ and $\delta \Phi^{a_\pi} = -2\,
\Lambda^{\eta^\pi \zeta^\pi} \epsilon^{a_\pi b_\pi}
\Phi^{b_\pi}$ where the gauge parameter $\Lambda^{\eta^\pi
  \zeta^\pi} = \eta_i^\pi \zeta_j^\pi \Lambda^{ij}$, under which
the derivative ${\mathcal{D}}$ transforms covariantly.  This can also
be fixed to set ${\mathcal{D}} = d$ on each $E_\pi$.  Notice that one
has precisely the right number of these gauge symmetries to fix all
the independent projections $\phi^{\eta^\pi \zeta^\pi}$ appearing
in the covariant derivatives.

After doing this one combines \eqref{eq:pot2proj} and \eqref{eq:intAeta2} to write
\begin{equation}
\label{eq:LESYMgf}
\begin{aligned}
  \eL^E =& \sum_{\pi =1}^M -\half P_{IJ}^{\eta^\pi
    \zeta^\pi} \partial_\mu X_I^{a_\pi} \partial^\mu X_J^{a_\pi}
  -\half (\Delta_{\eta^\pi \zeta^\pi} )^2 P_{IJ}^{\eta^\pi
    \zeta^\pi} X_I^{a_\pi} X_J^{a_\pi}  \\
  &+\sum_{\pi =1}^M -\tfrac{1}{Y^{\zeta^\pi \zeta^\pi}} (
  2\, \partial_{[ \mu} A_{\nu ]}^{\zeta^\pi a_\pi} ) (
  2\, \partial^{[ \mu} A^{\nu ]\, \zeta^\pi a_\pi} ) -
  \tfrac{2}{Y^{\zeta^\pi \zeta^\pi}} (\Delta_{\eta^\pi
    \zeta^\pi} )^2 A_\mu^{\zeta^\pi a_\pi} A^{\mu\, \zeta^\pi
    a_\pi}~,
\end{aligned}
\end{equation}
describing precisely the bosonic part of the lagrangian for free
decoupled abelian $N=8$ massive vector supermultiplets on each
$E_\pi$, whose bosonic fields comprise the six scalars
$P_{IJ}^{\eta^\pi \zeta^\pi} X_J^{a_\pi}$ and gauge field $-2
\tfrac{1}{\| X^{\zeta^\pi} \|} A^{\zeta^\pi
  a_\pi}$, all with mass $\Delta_{\eta^\pi \zeta^\pi}$ on each
$E_\pi$.  It is worth stressing that we have presented
\eqref{eq:LESYMgf} as a sum over all $E_\pi$ just so that the masses
$\Delta_{\eta^\pi \zeta^\pi}$ on each factor can be written more
explicitly.  We could equally well have presented things in terms of a
sum over the equivalence classes $E_{[\pi]}$, as in
\eqref{eq:LbetaSYM}, whereby the relative proportionality constants
for the $\Delta_{\eta^\pi \zeta^\pi}$ within a given class
${[\pi]}$ would be absorbed into the definition of the corresponding
$J_{[\pi]}$.

The lagrangian on a given $E_\pi$ in the sum \eqref{eq:LESYMgf} can
also be obtained from the truncation of an $N=8$ super Yang--Mills
theory with euclidean gauge algebra $\fg$ via the procedure described
at the end of section~\ref{sec:reviewSYM}.  In particular, let us
identify a given $E_\pi$ with the Cartan subalgebra of a semisimple
Lie algebra $\fg$ of rank two.  Then we require $- \|
X^{\zeta^\pi} \|^2 \, ( \ad_y )^2 =
(\Delta_{\eta^\pi \zeta^\pi} )^2 \, {\bf 1}_2$ on $E_\pi$ for
some constant $y \in E_\pi^\perp \subset \fg$.  In this case $\fg$
must be either $\fsu(3)$, $\fso(5)$, $\fso(4)$ or $\fg_2$ and
$E_\pi^\perp$ is identified with the root space of $\fg$ whose
dimension is 6, 8, 4 or 12 respectively.  A solution in this case is to
take $y$ proportional to the vector with only +1/-1 entries along the
positive/negative roots of $\fg$.  The proportionality constant here
being $\frac{\Delta_{\eta^\pi \zeta^\pi}}{\sqrt{h ( \fg
    )} \| X^{\zeta^\pi} \|}$ where $h( \fg )$ is the
dual Coxeter number of $\fg$ and equals 3, 3, 2 or 4 for
$\fsu(3)$, $\fso(5)$, $\fso(4)$ or $\fg_2$ respectively (it is
assumed that the longest root has norm-squared equal to 2 with
respect to the Killing form in each case).

Recall from \cite{LambertTong} that several of these rank two Lie
algebras are thought to correspond to the gauge algebras for $N=8$
super Yang--Mills theories whose IR superconformal fixed points
are described by the Bagger--Lambert theory based on $S_4$ for two
M2-branes on $\RR^8 / \ZZ_2$ (with Lie algebras $\fso(4)$,
$\fso(5)$ and $\fg_2$ corresponding to Chern--Simons levels
$k=1,2,3$).  It would interesting to understand whether there is
any relation with the aforementioned truncation beyond just
numerology! The general mass formulae we have obtained are
somewhat reminiscent of equation (26) in \cite{LambertTong} for
the BLG model based on $S_4$ which describes the mass in terms of
the area of the triangle formed between the location of the two
M2-branes and the orbifold fixed point on $\RR^8 / \ZZ_2$.  More
generally, it would be interesting to understand whether there is
a specific D-brane configuration for which $\eL^E$ is the
low-energy effective lagrangian?

\subsubsection{$\eL^{E_0}$}
\label{sec:L-Ezero}

The field $A^{i a_0}$ appears only linearly in one term in
$\eL_{\text{CS}}^0$ and is therefore a Lagrange multiplier imposing
the constraint $K_{ijk a_0} A^{jk} = d \gamma_{i a_0}$, where
$\gamma_{i a_0}$ is a scalar field on $\RR^{1,2}$ taking values in
$\RR^r \otimes E_0$.

Substituting this condition into the lagrangian allows us to write
\begin{equation}
\label{eq:intgamma}
\begin{aligned}
  -\half D_\mu X_I^{a_0} D^\mu X_I^{a_0} + \eL_{\text{CS}}^{E_0} =&
  -\half \partial_\mu \left( X_I^{a_0} - \gamma_i{}^{a_0} X_I^i
  \right) \partial^\mu \left( X_I^{a_0} - \gamma_j{}^{a_0} X_I^j
  \right) \\
  &- \tfrac{1}{3} A^{ij} \wedge d \gamma_{i a_0} \wedge d
  \gamma_{ja_0} + \half L_{ijkl} A^{ij} \wedge d A^{kl}~.
\end{aligned}
\end{equation}
The first line shows that we can simply redefine the scalars
$X_I^{a_0}$ such that they decouple and do not interact with any other
fields in the theory.

Notice that none of the projections $A^{\eta^\pi \zeta^\pi} = d
\phi^{\eta^\pi \zeta^\pi}$ of $A^{ij} $ that appeared in $\eL^E$
can appear in the second line of \eqref{eq:intgamma} since the
corresponding terms would be total derivatives.  Consequently, our
indifference to $\eL^{E_0}$ in the gauge-fixing that was described for
$\eL^E$, resulting in \eqref{eq:LESYMgf}, was indeed
legitimate.  Furthermore, there can be no components of $A^{ij}$ along
the 2-planes in $\RR^r$ spanned by the nonanishing components of
$K_{ijk a_0}$ here for the same reason.

The contribution coming from the Chern--Simons term in the second line
of \eqref{eq:intgamma} is therefore completely decoupled from all the
other terms in the lagrangian.  It has a rather unusual-looking
residual gauge symmetry, inherited from that in the original
Bagger--Lambert theory, under which $\delta \gamma_{i a_0} = \sigma_{i
  a_0} := K_{i a_0 kl} \Lambda^{kl}$ and $L_{ijkl} \left( \delta
  A^{kl} - d \Lambda^{kl} \right) = \sigma_{[i}{}^{a_0} d \gamma_{j]
  a_0}$ for any gauge parameter $\Lambda^{ij}$.  In addition to the
second line of \eqref{eq:intgamma} being invariant under this gauge
transformation, one can easily check that so is the tensor $L_{ijkl} d
A^{kl} - d \gamma_{i a_0} \wedge d \gamma_{j a_0}$.  This is perhaps
not surprising since the vanishing of this tensor is precisely the
field equation resulting from varying $A^{ij}$ in the second line of
\eqref{eq:intgamma}.  The important point though is that this
gauge-invariant tensor is exact and thus the field equations resulting
from the second line of \eqref{eq:intgamma} are precisely equivalent
to those obtained from an abelian Chern--Simons term for the gauge
field $C_{ij} := L_{ijkl} A^{kl} - \gamma_{[i}{}^{a_0} \wedge d
\gamma_{j] a_0}$ (where the $[ij]$ indices do not run over any
2-planes in $\RR^r$ which are spanned by the nonvanishing components
of $\eta_{[i}^\pi \zeta_{j]}^\pi$ and $K_{ijk a_0}$).

In summary, up to the aforementioned field redefinitions, we have found that
\begin{equation}
\label{eq:intgamma2}
\eL^{E_0} = -\half \partial_\mu X_I^{a_0} \partial^\mu X_I^{a_0} + \half M^{ijkl} C_{ij} \wedge d C_{kl}~,
\end{equation}
for some constant tensor $M^{ijkl}$, which can be taken to obey
$M^{ijkl} = M^{[ij][kl]} = M^{klij}$, that is generically a
complicated function of the components $L_{ijkl}$ and $K_{ijk
  a_0}$.  Clearly this redefined abelian Chern--Simons term is only
well-defined in a path integral provided the components $M^{ijkl}$ are
quantised in suitable integer units.  However, since none of the
dynamical fields are charged under $C_{ij}$ then we conclude that the
contribution from $\eL^{E_0}$ is essentially trivial.

\subsection{Examples}
\label{sec:BLexamples}

Let us end by briefly describing an application of this formalism to describe the unitary gauge theory resulting from the
Bagger--Lambert theory associated with two of the admissible index-$2$ 3-Lie algebras in the IIIb family from \cite{2p3Lie} that
were detailed in section~\ref{sec:examples}.

\subsubsection{$V_{\text{IIIb}}(0,0,0,\fh,\fg,\psi)$}
\label{sec:noJ}

The data needed for this in Theorem~\ref{thm:main} is $\kappa |_{\fh} = (0,1)^t$, $\kappa |_{\fg_{\alpha}}= (\psi_\alpha, 1)^t$.
The resulting Bagger--Lambert lagrangian will only get a contribution from $\eL^W$ and describes a sum of separate $N=8$ super
Yang--Mills lagrangians on $\fh$ and on each factor $\fg_\alpha$, with the respective euclidean Lie algebra structures describing
the gauge symmetry.  The super Yang--Mills theory on $\fh$ has coupling $\| X^{u_2} \|$ and the seven scalar fields occupy the
hyperplane orthogonal to $X^{u_2}$ in $\RR^8$.  Similarly, the $N=8$ theory on a given $\fg_\alpha$ has coupling $\| \psi_\alpha
X^{u_1} + X^{u_2} \|$ with scalars in the hyperplane orthogonal to $\psi_\alpha X^{u_1} + X^{u_2}$.  This is again generically a
super Yang--Mills theory though it degenerates to a maximally supersymmetric free theory for all eight scalars if there are any
values of $\alpha$ for which $\psi_\alpha X^{u_1} + X^{u_2} =0$.

\subsubsection{$V_{\text{IIIb}}(E,J,0,\fh,0,0)$}
\label{sec:soJ}

The data needed for this in Theorem~\ref{thm:main} is $\kappa |_{\fh} = (0,1)^t$ and $J^\pi = \eta^\pi \wedge \zeta^\pi$ where
$\eta^\pi$ and $\zeta^\pi$ are 2-vectors spanning $\RR^2$ for each value of $\pi$ and $E = \bigoplus_{\pi =1}^M E_\pi$.  The data
comprising $J^\pi$ can also be understood as a special case of a general admissible index $r$ 3-Lie algebra having all $\eta^\pi
\wedge \zeta^\pi$ spanning the same 2-plane in $\RR^r$ (when $r=2$ this is unavoidable, of course).  The resulting Bagger--Lambert
lagrangian will get one contribution from $\eL^W$, describing precisely the same $N=8$ super Yang--Mills theory on $\fh$ we saw
above, and one contribution from $\eL^E$.  The latter being the simplest case of the lagrangian \eqref{eq:LbetaSYM} where there is
just one equivalence class of 2-planes spanned by all $\eta^\pi \wedge \zeta^\pi$ and the gauge symmetry is based on the
lorentzian Lie algebra $\fd (E, \RR )$.  The physical degrees of freedom describe free abelian $N=8$ massive vector
supermultiplets on each $E_\pi$ with masses $\Delta_{\eta^\pi \zeta^\pi}$ as in \eqref{eq:LESYMgf}.  Mutatis mutandis, this
example is equivalent to the Bagger--Lambert theory resulting from the most general finite-dimensional 3-Lie algebra example
considered in section 4.3 of \cite{Ho:2009nk}.

\section*{Acknowledgments}

EME would like to thank Iain Gordon for useful discussions.  PdM is
supported by a Seggie-Brown Postdoctoral Fellowship of the School of
Mathematics of the University of Edinburgh.

\bibliographystyle{utphys}
\bibliography{AdS,AdS3,ESYM,Sugra,Geometry,Algebra}

\end{document}